\documentclass[12]{article}
\usepackage[utf8]{inputenc}
\usepackage{comment}
\usepackage{graphicx} 
\usepackage{booktabs}
\usepackage[left=1in, bottom=1in, right=1in, top=1in]{geometry}
\usepackage{amsmath, bm}
\usepackage{amsthm}
\usepackage{setspace}
\usepackage{algorithm}
\usepackage{algpseudocode}
\usepackage{amsfonts}
\usepackage{subcaption}
\usepackage{hyperref}
\usepackage{authblk}
\usepackage[dvipsnames]{xcolor}
\usepackage{enumitem}
\title{Analysing Models for Volatility Clustering \\ with Subordinated Processes: VGSA and Beyond}
\author{
  Sourojyoti Barick\thanks{Interdisciplinary Statistical Research Unit, Indian Statistical Institute, Kolkata, India}
  \and
  Sudip Ratan Chandra\thanks{Independent Researcher, India}
}
\date{\today}

\newtheorem{theorem}{Theorem}[section]

\newtheorem{lemma}[theorem]{Lemma}

\theoremstyle{definition}
\newtheorem{definition}{Definition}[section]

\theoremstyle{remark}

\usepackage{dsfont}

\newcommand{\E}{\mathbb{E}}

\newcommand{\fst}[1]{{\left(#1 \right)}}
\newcommand{\secnd}[1]{{\left\{#1 \right\}}}
\newcommand{\thrd}[1]{{\left[#1 \right]}}

\newcommand{\distconv}{\xrightarrow{\mathcal{D}}}

\hypersetup{
	colorlinks=true,
	linkcolor=red,
	filecolor=magenta,      
	urlcolor=red,
	citecolor=blue
}
\usepackage[citestyle=authoryear-comp,
    maxcitenames=2,
    giveninits=true,
    uniquename=init,
    uniquelist=false,
	ibidtracker=context,
    bibstyle=authoryear,
    dashed=false, 
    date=year,
	sorting=nyt,
    minbibnames=3,
	hyperref=true,
	backref=true,
	citecounter=true,
	citetracker=true,
    natbib=true, 
	backend=biber, 
]{biblatex}

\DeclareFieldFormat{citehyperref}{%
  \DeclareFieldAlias{bibhyperref}{noformat}
  \bibhyperref{#1}}

\DeclareFieldFormat{textcitehyperref}{%
  \DeclareFieldAlias{bibhyperref}{noformat}
  \bibhyperref{%
    #1%
    \ifbool{cbx:parens}
      {\bibcloseparen\global\boolfalse{cbx:parens}}
      {}}}

\savebibmacro{cite}
\savebibmacro{textcite}

\renewbibmacro*{cite}{%
  \printtext[citehyperref]{%
    \restorebibmacro{cite}%
    \usebibmacro{cite}}}

\renewbibmacro*{textcite}{%
  \ifboolexpr{
    ( not test {\iffieldundef{prenote}} and
      test {\ifnumequal{\value{citecount}}{1}} )
    or
    ( not test {\iffieldundef{postnote}} and
      test {\ifnumequal{\value{citecount}}{\value{citetotal}}} )
  }
    {\DeclareFieldAlias{textcitehyperref}{noformat}}
    {}%
  \printtext[textcitehyperref]{%
    \restorebibmacro{textcite}%
    \usebibmacro{textcite}}}

\DeclareAutoCiteCommand{inline}{\textcite}{\parencite}

\begin{document}
\maketitle

\begin{abstract}
This paper explores a comprehensive class of time-changed stochastic processes constructed by subordinating Brownian motion with Lévy processes, where the subordination is further governed by stochastic arrival mechanisms such as the Cox–Ingersoll–Ross (CIR) and Chan–Karolyi–Longstaff–Sanders (CKLS) processes. These models extend classical jump frameworks like the Variance Gamma (VG) and CGMY processes, allowing for more flexible modeling of market features such as jump clustering, heavy tails, and volatility persistence. We first revisit the theory of Lévy subordinators and establish strong consistency results for the VG process under Gamma subordination. Building on this, we prove asymptotic normality for both the VG and VGSA (VG with stochastic arrival) processes when the arrival process follows CIR or CKLS dynamics. The analysis is then extended to the more general CGMY process under stochastic arrival, for which we derive analogous consistency and limit theorems under positivity and regularity conditions on the arrival process. A simulation study accompanies the theoretical work, confirming our results through Monte Carlo experiments, with visualizations and normality testing (via Shapiro-Wilk statistics) that show approximate Gaussian behavior even for processes driven by heavy-tailed jumps. This work provides a rigorous and unified probabilistic framework for analyzing subordinated models with stochastic time changes, with applications to financial modeling and inference under uncertainty.
\end{abstract}


\vspace{1em}

\section{Introduction}
Financial markets often exhibit price and rate movements that are not continuous but instead involve sudden jumps to new levels. Such discontinuities are evident in the pricing of options, where market data reflects these abrupt changes. Experts in the field have highlighted the limitations of pure diffusion-based models in explaining phenomena such as the pronounced smile effect observed in short-dated option prices. Consequently, significant efforts have been devoted to developing models that incorporate price jumps, with Poisson-type jump components in jump-diffusion models being a key solution to these challenges. The variance gamma (VG) process is a widely recognized Lévy process extensively employed in financial modeling. It is a pure jump process defined by high activity, consistent with the normal distribution, permitting an infinite number of jumps within any time interval. 

The VG process was first introduced in financial domain by \citet{Madan_Seneta}, where it was described as analogous to the model proposed by \citet{Paretz_1972}, as it is obtained by mixing the normal distribution on the variance parameter. The Variance Gamma (VG) distribution can be interpreted as observing data \( X \) originating from a normal distribution, where the variance parameter of the normal distribution follows a gamma prior distribution. 
The VG process is a pure jump process with finite variation and can be expressed as the difference of two gamma processes, one capturing price increases and the other accounting for price decreases. The probability density function (PDF) of the VG process (see \cite{Pearson1929ONTD,BIBBY2003211,kotz2001laplace}) involves the modified Bessel function and is governed by three key parameters, as depicted by the density function at time \( t \):  
\begin{align*}
f_t(x ; \sigma, \nu, \theta) & = \int_0^{\infty} normal\left(\theta g, \sigma^2 g\right) \times \operatorname{gamma}\left(\frac{t}{\nu}, \nu\right) d g \\
& = \int_0^{\infty} \frac{1}{\sigma \sqrt{2 \pi g}} \exp \left(-\frac{(x-\theta g)^2}{2 \sigma^2 g}\right) \frac{g^{t / \nu-1} e^{-g / \nu}}{\nu^{t / \nu} \Gamma(t / \nu)} d g.
\end{align*}  Due to its flexibility, the VG process is often well-suited for statistical modeling in financial markets. In \citet{only_seneta}, the VG distribution was fitted to financial market data in a more general asymmetric case, capturing real-world market behavior effectively. Additionally, in \citet{madan_car_chang_eric} and \citet{Madan_Frank}, the VG process was utilized to develop option pricing models, demonstrating its applicability across diverse financial market scenarios. Another notable feature of the VG process is its ability to include the lognormal density and the Black–Scholes (\citet{black_scholes}) formula as parametric special cases.

The VG process can also be characterized as a Lévy process, which plays a fundamental role in various scientific disciplines, such as turbulence studies, quantum field theory, network analysis, financial markets, and more (see \cite{Applebaum_2004, Sato, kyprianou2006introductory}). In the field of mathematical finance, Lévy processes have gained substantial attention for their ability to better capture the observed dynamics of financial markets compared to traditional models based solely on Brownian motion. The books of \citet{tankov2015financial} and \citet{schoutens2003levy} extensively explore the applications of Lévy processes in mathematical finance. 

In addition to the VG process, another widely studied Lévy process is the CGMY process, introduced in \citet{cgmy}. The CGMY process serves as a generalization of Kou's jump-diffusion model (\cite{kou_jump}) and the VG process, offering greater flexibility in modeling financial market dynamics. Furthermore, the CGMY process is a specific case of the more general Kobol process, which has been extensively analyzed in the works of \citet{kobol}. The CGMY process is particularly notable for its ability to capture both finite and infinite activity jumps, making it highly suitable for a broad range of applications in financial mathematics.

The previously discussed models often struggle to capture volatility clustering, a critical feature of financial markets where periods of high or low volatility tend to persist. This limitation is typically addressed by introducing random time changes, which allow the time scale to evolve stochastically. Volatility clustering is effectively modeled when the rate of time change is mean-reverting, as demonstrated by the Cox–Ingersoll–Ross (CIR) process (\cite{cox1985}), which ensures that extreme deviations revert to equilibrium over time. To incorporate this feature into the Variance Gamma framework, the Variance Gamma with Stochastic Arrival (VGSA) model was proposed in \citet{svlp}. This model introduces stochastic time changes into the Variance Gamma process, significantly enhancing its ability to capture the temporal dependencies and clustering patterns commonly observed in financial markets. In addition to VGSA, \citet{svlp} also introduced two other stochastic volatility Lévy processes: the NIGSV and CGMYSV models, further expanding the toolkit for modeling stochastic volatility in financial applications.

Even though it is used extensively, the explicit distributional properties of the VGSA process are not clearly understood. 
In this paper, we aim to explicitly derive the mean and variance of the Variance Gamma with Stochastic Arrival (VGSA) model, providing an explicit understanding of its distributional properties. We will also explore the asymptotic behavior of the VGSA model, which reveals important insights into its long-term characteristics and dynamics. Additionally, these derivations will serve as a foundation for examining more general cases, enabling us to study broader asymptotic properties that extend beyond the specific VGSA framework. To complement the theoretical results, we will also perform simulations to validate and illustrate the derived properties, offering a comprehensive analysis of the VGSA model's behavior under various conditions.
\subsection{Plan of the Paper}

This paper is organized as follows. We begin with a brief overview of L\'evy processes in Section~\ref{sec:prem_theo}, highlighting fundamental properties relevant to modeling jump-driven financial dynamics. This includes definitions and examples of commonly used subordinators, such as the Gamma and CGMY processes. We also introduce several stochastic arrival (SA) processes and establish their strict positivity properties.

In Section~\ref{sec:VGSA}, we focus on establishing the strong consistency of subordinators and their associated subordinated processes. We begin with the Variance Gamma (VG) process and prove its strong consistency when constructed using a deterministic Gamma subordinator, which forms the foundation for analyzing more complex time-changed models. We further derive the asymptotic normality of the VG process. Subsequently, we analyze the VGSA process, rigorously proving both its strong consistency and asymptotic normality. Additionally, we present theoretical results for the CIR and CKLS processes that underpin these constructions.

In the following section~\ref{sec:SBSA}, we examine subordinated Brownian motions of the form $X(t) = \theta S(t) + \sigma W(S(t))$, where $S(t)$ is a L\'evy subordinator and $W(t)$ is a standard Brownian motion. We derive asymptotic properties for this class of processes.
We then extend this framework by introducing stochastic arrival mechanisms governed by positive-valued processes such as the Cox–Ingersoll–Ross (CIR) and the Chan–Karolyi–Longstaff–Sanders (CKLS) processes. We denote this broader class as L\'evy processes subordinated by SA processes. For these, we establish sufficient conditions for strong consistency and asymptotic normality, extending our theoretical results to include CGMY-type jump processes.

Section~\ref{sec:simu_res} presents an extensive simulation study to support our theoretical findings. This includes numerical illustrations demonstrating strong consistency and distributional convergence for the VGSA and CGMY-SA processes under both CIR and CKLS stochastic clocks. Histograms, kernel density estimates, and normality tests are used to visualize the effects of stochastic arrival on the distributional behavior of the subordinated processes.

We conclude by highlighting the theoretical and practical implications of our results, and suggest directions for future research including parameter estimation, model calibration, and financial applications. All extended proofs are provided in Appendix for completeness.

\section{Preliminaries}\label{sec:prem_theo}

This section provides an introduction to Lévy processes, highlighting essential definitions and properties that form the basis for their use in stochastic modeling. For a detailed exposition, readers are referred to \citet{Sato, hirsa2024computational}. We consider a stochastic process $L = \{L_t\}_{t \geq 0}$, defined on a filtration $\left(\Omega, \{\mathcal{F}_t\}_{t \geq 0}, \mathcal{F}, \mathbb{P}\right)$, and adapted to the filtration $\{\mathcal{F}_t\}_{t \geq 0}$.

\begin{definition}[Lévy Process, \cite{tankov2015financial,Marshall_Thesis}]
    A Lévy process is a real-valued, adapted, càdlàg process (i.e., its sample paths are almost surely right-continuous with left limits) satisfying $L_0 = 0$ almost surely, and the following properties:
    \begin{itemize}
        \item \textbf{L1 (Independent increments):} For any $0 \leq s < t \leq T$, $L_t - L_s$ is independent of $\mathcal{F}_s$.
        \item \textbf{L2 (Stationary increments):} For any $s, t \geq 0$, the distribution of $L_{t+s} - L_t$ depends only on $s$.
        \item \textbf{L3 (Stochastic continuity):} For every $t \geq 0$ and $\epsilon > 0$, 
        \[
        \lim_{s \to t} \mathbb{P}\left(\left|L_t - L_s\right| > \epsilon\right) = 0.
        \]
    \end{itemize}
\end{definition}

\begin{definition}[Infinite Divisibility, \citet{tankov2015financial}]
    A probability distribution $G$ on $\mathbb{R}$ is called infinitely divisible if, for any integer $m \geq 2$, there exist $m$ independent and identically distributed (i.i.d.) random variables $X_1, \ldots, X_m$ such that the sum $X_1 + \cdots + X_m$ has distribution $G$.
\end{definition}

Lévy processes are closely connected to the concept of infinite divisibility. The following lemma formalizes this relationship.
\begin{lemma}[Infinite Divisibility and Lévy Processes, \cite{tankov2015financial}]
    Let $\{L_t\}_{t \geq 0}$ be a Lévy process. Then, for any $t \geq 0$, $L_t$ follows an infinitely divisible distribution. Conversely, if $G$ is an infinitely divisible distribution, there exists a Lévy process $\{L_t\}$ such that the distribution of $L_1$ is $G$.
\end{lemma}

\begin{proof}
    See \citet{tankov2015financial}, Proposition 3.1 for the detailed proof.
\end{proof}

A fundamental result for Lévy processes is the Lévy-Khintchine formula, which provides the characteristic function of a Lévy process and links its structure to its characteristic exponent.

\begin{theorem}[Lévy-Khintchine Formula, \cite{schoutens2003levy}]\label{theorem:levy_khint}
    For a Lévy process $\{L_t\}_{0 \leq t \leq T}$, the characteristic function is given by:
    \begin{align}
        \mathbb{E}\left[\mathrm{e}^{i u L_t}\right] = \mathrm{e}^{t \psi(u)} = \exp\left[t\left(i b u - \frac{c u^2}{2} + \int_{\mathbb{R}} \left(\mathrm{e}^{i u x} - 1 - i u x \mathbf{1}_{\{|x| < 1\}}\right) v(\mathrm{d}x)\right)\right],\label{eqn:levy_khint}
    \end{align}
    where $\psi(u)$ is the characteristic exponent, \( b \in \mathbb{R} \) is the drift term, \( c \geq 0 \) represents the Gaussian component, and \( v \) is the Lévy measure, which satisfies:
    \[
    \int_{\mathbb{R} \setminus \{0\}} (1 \wedge x^2) v(\mathrm{d}x) < \infty.
    \]
\end{theorem}

\begin{proof}
    See \citet{papa_2008}, Theorem 4.9 for the detailed proof.
\end{proof}

Every infinitely divisible distribution corresponds to a Lévy triplet \(\fst{b, c, v\mathrm{d}x)}\).
\begin{theorem}[See \cite{statland1965local}]\label{thm:levy_strong}
    Let L be a Lévy process  with triple $(b, c, v)$. Then, \[p:=\mathrm{P}\left\{\lim \sup _{t \rightarrow \infty}\left\|L_t / t\right\|<\infty\right\}\] is zero or one. And $p=1$ if and only if $\int_{\mathbb{R}}|x|\mathbf{1}_{\secnd{|x| \geq 1}}  v(\mathrm{~d} x)<\infty$. And when $p=1$,

$$
\lim _{t \rightarrow \infty} \frac{L_t}{t}=b+\int_{\mathbb{R}}x\mathbf{1}_{|x| \geq 1}  v(\mathrm{~d} x) \quad \text { a.s. }
$$
\end{theorem}

 We now focus on a specific class of Lévy processes. A subordinator is a special type of Lévy process that is non-decreasing and takes values in \([0, \infty)\).

\begin{definition}[Subordinator, \cite{papa_2008}]
    A subordinator is a Lévy process that is almost surely non-decreasing. For a Lévy process \(L_t\) to be a subordinator, its Lévy triplet must satisfy the following conditions: 
    \[
    v(-\infty, 0) = 0, \quad c = 0, \quad \int_{(0,1)} x v(\mathrm{d}x) < \infty, \quad \text{and} \quad \gamma = b - \int_{(0,1)} x v(\mathrm{d}x) > 0.
    \]
    Under these conditions, the Lévy-Khintchine formula for a subordinator simplifies to:
    \begin{align}
        \mathbb{E}\left[\mathrm{e}^{i u L_t}\right] = \exp\left[t\left(i u \gamma + \int_{(0, \infty)}\left(\mathrm{e}^{i u x} - 1\right) v(\mathrm{d}x)\right)\right]. \label{eqn:sub_levy_char} 
    \end{align}
\end{definition}


The Variance Gamma (VG) process serves as an excellent example of a process built using a subordinator, specifically a Brownian motion evaluated at a Gamma subordinator. The VG process can be expressed as:
\begin{align}
  VG(t) = \theta G(t) + \sigma W(G(t)),\label{eqn:VG_eqn}  
\end{align}

where \( G(t) \sim \Gamma\left(\frac{t}{\nu}, \nu\right) \) is a Gamma subordinator, and \( W(t) \) is a standard Wiener process. The characteristic function of the VG process is given by:
\begin{align}
    \mathbb{E}\left(e^{i u VG(t)}\right) = \left(\frac{1}{1 - i u \theta \nu + \frac{\sigma^2 u^2 \nu}{2}}\right)^{\frac{t}{\nu}}.\label{eqn:vg_char}
\end{align}

The derivation of the characteristic function for the Variance Gamma (VG) process relies on its conditional distribution and the tower property of expectations (see \cite{hirsa2024computational}). A natural generalization of the VG process is the CGMY process, which belongs to the class of pure-jump Lévy models widely used in financial markets. The CGMY process enhances flexibility in modeling jump dynamics and is regarded as one of the most well-structured Lévy processes for financial applications. It is defined through its Lévy measure as follows:
\begin{align}
   \nu(dx) = C \cdot \frac{e^{-G|x|}}{|x|^{1+Y}} \mathbf{1}_{x<0} \, dx + C \cdot \frac{e^{-M x}}{x^{1+Y}} \mathbf{1}_{x>0} \, dx,\label{eqn:cgmy_eqn} 
\end{align}
where:
\begin{itemize}
    \item $C > 0$ is the overall activity level,
    \item $G, M > 0$ are exponential tempering parameters controlling left and right tails, respectively,
    \item $Y < 2$ governs the activity of small jumps.
\end{itemize}
The CGMY process, introduced by \citet{cgmy}, has become a cornerstone in the modeling of asset returns, particularly in the context of option pricing and quantitative finance. Structurally, it shares the same representation as the Variance Gamma (VG) process, namely \( \mathrm{CGMY}(t) = \theta S(t) + \sigma W(S(t)) \), where \( W(\cdot) \) is a standard Brownian motion and \( S(t) \) is an increasing Lévy process acting as the stochastic clock. While the VG process uses a Gamma subordinator, the CGMY model generalizes this by employing a more flexible Lévy subordinator capable of capturing a broader range of jump behaviors (see \citet{madan2006representing}).

Despite its versatility, the standard CGMY model falls short in capturing temporal market features such as trade clustering and stochastic volatility. To overcome this, the CGMY-SA framework was proposed (see \cite{svlp}), where the Lévy process is further time-changed by a stochastic arrival (SA) process—typically a positive-valued diffusion like the CIR or CKLS model. This added layer of stochasticity allows the model to account for time-varying trading intensity and volatility clustering, making it more suitable for high-frequency and empirical financial data. In this paper, we focus on such time-changed CGMY models and analyze their statistical and asymptotic properties.

Before delving into the CGMY-SA framework, we first introduce candidate models for the stochastic arrival process. One of the most widely used mean-reverting SA processes is the Cox-Ingersoll-Ross (CIR) model, which has been employed to incorporate clustering effects in pure-jump Lévy models, as discussed in \cite{svlp}. For the CIR process to serve as a valid stochastic clock in time-changed models, it must remain strictly positive to ensure that the integrated process is well-defined. Specifically, the time-change is introduced via the integrated CIR (ICIR) process, given by:
\begin{align}
    T(t) = \int_0^t y(u) \, du,\label{eqn:icir_eqn}
\end{align}
where \( y(u) \) evolves according to the CIR dynamics:
\begin{align}
    d y(u) = \kappa (\eta - y(u)) \, du + \lambda \sqrt{y(u)} \, dW_u. \label{eqn:cir_eqn}
\end{align}
Under the Feller condition \( 2\kappa\eta > \lambda^2 \), the boundary at zero becomes inaccessible, ensuring that \( y(u) \) remains strictly positive almost surely (see \cite{andersenPit, Cox1976, Karlin1981ASC}). However, to explore a more general class of stochastic arrival (SA) processes, one can consider the Chan–Karolyi–Longstaff–Sanders  (CKLS) model (see \cite{ckls1992}), which extends the CIR dynamics and is defined by the following stochastic differential equation:
\begin{align}
    d y(u) = \kappa (\eta - y(u)) \, du + \lambda y(u)^\alpha \, dW_u. \label{eqn:ckls_diff_eqn}
\end{align}
If we restrict the exponent to \( \alpha \in (0.5, 1] \), then the process remains strictly positive under suitable parameter conditions. The corresponding integrated process, which serves as the stochastic clock, is given by:
\begin{align}
    \mathcal{Y}(t) = \int_0^t y(u) \, du. \label{eqn:iergodic}
\end{align}
Thus, the CKLS process provides a flexible and theoretically sound alternative to the CIR process for modeling stochastic arrival times in subordinated Brownian motion or Lévy frameworks.

\section{Main Theoretical Results on the VGSA Framework}\label{sec:VGSA}
To understand the construction of the Variance Gamma with Stochastic Arrival (VGSA) process, it is essential to recognize that the Gamma process \( G(t) \) is an almost surely increasing Lévy process. In the standard Variance Gamma (VG) model, this Gamma process is typically parameterized as \( G(t) \sim \Gamma\left(\frac{t}{\nu}, \nu\right) \), where \( \nu \) controls the variance of the time-change. To introduce greater flexibility and accommodate features such as volatility clustering observed in financial data, one can generalize this framework by replacing the deterministic time index \( t \) with a stochastic, yet monotone, time-change process \( T(t) \), resulting in a subordinated Gamma process \( G(T(t)) \sim \Gamma\left(\frac{T(t)}{\nu}, \nu\right) \). 

For this construction to be well-defined, the time-change process \( T(t) \) must also be almost surely increasing. While this can be trivially satisfied by choosing a deterministic, strictly increasing function \( f(t) \), a more natural and empirically motivated choice is to use a stochastic process that is almost surely monotone—such as the integrated Cox–Ingersoll–Ross (ICIR) process. The ICIR process introduces a mean-reverting mechanism into the arrival rate of trading activity and ensures positivity under the Feller condition \( 2\kappa\eta > \lambda^2 \). The resulting VGSA process is then given by:
\begin{align}
    VGSA(t) = \theta G(T(t)) + \sigma W(G(T(t))), \quad \text{where} \quad G(T(t)) \sim \Gamma\left(\frac{T(t)}{\nu}, \nu\right), \label{eqn:vgsa_eqn}
\end{align}
with \( T(t) \) defined as the integrated CIR process in equations~\eqref{eqn:cir_eqn} and~\eqref{eqn:icir_eqn}.

This formulation preserves the essential property of monotonicity required for subordinating the Gamma process, while simultaneously enriching the model with stochastic features capable of capturing mean reversion and temporal clustering—two stylized facts commonly observed in financial time series. As a result, the VGSA process retains the pure-jump nature of the original Variance Gamma (VG) model while embedding a dynamic time-change structure through a stochastic arrival mechanism. This hybrid framework enables the model to better reflect the complex interplay between trading intensity and return variability. To understand how the integrated CIR (ICIR) process influences this stochastic time-change, it is essential to examine the statistical properties—particularly the mean and variance—of the ICIR process, which govern the temporal evolution and scaling behavior of the subordinated system.

Before establishing the strong consistency and asymptotic normality results for the VGSA process, we first revisit foundational properties of subordinator processes, the VG model, and the CIR/ICIR processes. These results will serve as building blocks for the subsequent asymptotic theory. Let us begin by considering a general subordinator process \( X = \{X_t\}_{t \geq 0} \), characterized by the Lévy–Khintchine representation in \eqref{eqn:sub_levy_char}. Assuming the first moment \( \mathbb{E}(X_1) \) exists, we have the following almost sure convergence:

\begin{theorem}\label{theorem:subordi_as}
    Let \( X = \{X_t\}_{t \geq 0} \) be a subordinator process with characteristic function as given in \eqref{eqn:sub_levy_char}, and suppose that \( \mathbb{E}(X_1) < \infty \). Then, the following strong law of large numbers holds:
    \begin{align}
        \lim_{t \to \infty} \frac{X_t}{t} \xrightarrow{a.s.} \mathbb{E}(X_1), \label{eqn:as_levy_sub}
    \end{align}
    where the expected value is given by:
    \begin{align*}
        \mathbb{E}(X_1) = \gamma + \int_{0}^\infty x \, v(dx),
    \end{align*}
    with \( \gamma \) denoting the drift coefficient and \( v(\cdot) \) the Lévy measure associated with the process.
\end{theorem}
A direct consequence of Theorem~\ref{theorem:subordi_as} yields the following lemma, which establishes the strong consistency of the Variance Gamma (VG) process:

\begin{lemma}\label{lemma:vg_as}
    Let \( VG(t) \) be as in~\eqref{eqn:VG_eqn}. Then, the process satisfies the following almost sure convergence:
    \begin{align}
        \lim_{t \to \infty} \frac{1}{t} VG(t) \xrightarrow{a.s.} \theta. \label{eqn:as_vg}
    \end{align}
\end{lemma}
In addition to the strong consistency established above, we also obtain the following central limit-type result that characterizes the asymptotic distribution of the Variance Gamma process:

\begin{lemma}\label{lemma:vg_clt}
    Let \( VG(t) \) be as in~\eqref{eqn:VG_eqn}. Then, the following distributional convergence holds:
    \begin{align}
        \frac{1}{\sqrt{t}} \left( VG(t) - t\theta \right) \distconv \mathcal{N} \left( 0, \nu\theta^2 + \sigma^2 \right), \label{eqn:vg_asymp}
    \end{align}
    where \( \distconv\) denotes convergence in distribution.
\end{lemma}

We now summarize key analytical results concerning the stochastic arrival (SA) process, focusing specifically on the Cox–Ingersoll–Ross (CIR) process and its integrated counterpart, the ICIR process.

\begin{lemma}\label{lemma:icir_mean_var}
    Let \( y(t) \) be a CIR process defined by~\eqref{eqn:cir_eqn}, and let the associated integrated process be \( T(t) \) (ICIR) as in~\eqref{eqn:icir_eqn}. Then, the mean and variance of \( T(t) \) can be expressed as:
    \begin{align}
        \mathbb{E}[T(t)] = u(t), \quad \mathrm{Var}[T(t)] = \lambda^2 w(t), \label{eqn:icir_mean_var}
    \end{align}
    where
    \begin{align}
        u(t) &= -\frac{y(0)}{\kappa} \left( e^{-\kappa t} - 1 \right) + \eta \left[ t + \frac{1}{\kappa} \left( e^{-\kappa t} - 1 \right) \right], \\
        w(t) &= \frac{y(0) - \eta}{\kappa^3} \left( -2\kappa t e^{-\kappa t} + 1 - e^{-2\kappa t} \right) 
        + \frac{\eta}{2\kappa^3} \left( 2\kappa t - 3 + 4e^{-\kappa t} - e^{-2\kappa t} \right),
        \label{eqn:u_w_t}
    \end{align}
    and \( y(0) \) is the initial value of the CIR process.
\end{lemma}

To further analyze the behavior of the ICIR process, we now examine its asymptotic properties, particularly focusing on the long-term growth rates of its mean and second moment.

\begin{lemma}\label{lemma:icir_as}
    Let \( T(t) = \int_0^t y(u) \, du \) denote the integrated CIR (ICIR) process, where \( y(u) \) follows the CIR dynamics~\eqref{eqn:cir_eqn}. Then, under the Feller condition \( 2\kappa\eta > \lambda^2 \), the ICIR process satisfies the following almost sure limits:
    \begin{align}
        \lim_{t \to \infty} \frac{1}{t} T(t) &\xrightarrow{a.s.} \eta, \\
        \lim_{t \to \infty} \frac{1}{t^2} T(t)^2 &\xrightarrow{a.s.} \eta^2. \label{eqn:as_icir}
    \end{align}
\end{lemma}
\begin{proof}
    From the ergodic property of the CIR process, we have:
    \begin{align*}
        \lim_{t \to \infty} \frac{1}{t} T(t) 
        = \lim_{t \to \infty} \frac{1}{t} \int_0^t y(u)\,du 
        \xrightarrow{a.s.} \eta.
    \end{align*}
    Now, let $f(x) = x^2$, which is continuous on $\mathbb{R}_+$. Applying the continuous mapping theorem to the above almost sure convergence, we get:
    \begin{align*}
        \lim_{t \to \infty} f\left( \frac{T(t)}{t} \right) 
        = \lim_{t \to \infty} \left( \frac{T(t)}{t} \right)^2 
        \xrightarrow{a.s.} f(\eta) = \eta^2.
    \end{align*}
    Hence proved.
\end{proof}

We now examine the distributional convergence of the ICIR process, which serves as a crucial step toward understanding the asymptotic distribution of the VGSA process. This result provides insights into the fluctuations of the stochastic clock driving the time-changed Gamma process.

\begin{lemma}\label{lemma:ICIR_asymp}
    For the integrated CIR (ICIR) process \( T(t) \) ~\eqref{eqn:icir_eqn}, under the Feller condition \( 2\kappa\eta > \lambda^2 \), the following distributional convergence holds:
    \begin{align}
        \frac{1}{\sqrt{t}}\left( T(t) - \eta t \right) \distconv \mathcal{N} \left(0, \frac{\eta \lambda^2}{\kappa^2} \right), \label{eqn:icir_dist_conv}
    \end{align}
    where \( \distconv \) denotes convergence in distribution.
\end{lemma}

\begin{proof}
We first note that the ICIR process \( T(t) \) can be decomposed as:
\[
T(t) = u(t) + \lambda I(t),
\]
where \( u(t) = \mathbb{E}[T(t)] \) is the deterministic part (given in Equation~\eqref{eqn:u_w_t}), and \( I(t) \) is a centered diffusion term. It can be shown that \( I(t) \) satisfies the stochastic integral representation:
\[
I(t) = \frac{e^{-\kappa t}}{\kappa} \int_0^t \left(e^{\kappa t} - e^{\kappa z}\right) \sqrt{y(z)} \, dW(z).
\]
The quadratic variation of \( \lambda I(t) \) is given by:
\[
\langle \lambda I \rangle_t = \lambda^2 \cdot \frac{e^{-2\kappa t}}{\kappa^2} \int_0^t \left(e^{\kappa t} - e^{\kappa z}\right)^2 y(z) \, dz.
\]
Using the ergodic property of the CIR process, we obtain:
\[
\frac{1}{t} \langle \lambda I \rangle_t \xrightarrow{a.s.} \frac{\eta \lambda^2}{\kappa^2}.
\]
Therefore, by the Martingale Central Limit Theorem, it follows that:
\[
\frac{\lambda I(t)}{\sqrt{t}} \distconv \mathcal{N} \left(0, \frac{\eta \lambda^2}{\kappa^2} \right).
\]
To complete the proof, observe that:
\[
\frac{u(t) - \eta t}{\sqrt{t}} \to 0 \quad \text{as} \quad t \to \infty,
\]
since \( u(t) \sim \eta t + o(\sqrt{t}) \). Hence, by Slutsky’s theorem, the result in Equation~\eqref{eqn:icir_dist_conv} follows.
\end{proof}
To establish the asymptotic properties of the VGSA process, we begin by computing its first two moments, namely the mean and variance. These moments provide insight into the process's behavior and are crucial for later deriving consistency and distributional convergence results.

\begin{lemma}[Moments of the VGSA Process]\label{lemma:vgsa_moments}
     The mean and variance of \( VGSA(t) \)~\eqref{eqn:vgsa_eqn} are given by:
    \begin{align}
    \begin{aligned}
        \mathbb{E}\left[VGSA(t)\right] &= \theta\, u(t), \\
        \mathrm{Var}\left[VGSA(t)\right] &= (\theta^2 \nu + \sigma^2) u(t) + \theta^2 \lambda^2 w(t),
    \end{aligned}
    \label{eqn:vgsa_mean_var}
    \end{align}
    where \( u(t) \) and \( w(t) \) are the mean and scaled variance of the ICIR process as defined in Equation~\eqref{eqn:u_w_t}.
\end{lemma}

\begin{proof}
    The VGSA process can be written in terms of the time-changed Gamma process as:
    \[
    VGSA(t) = \theta\, G(T(t)) + \sigma\, \sqrt{G(T(t))}\cdot \frac{W(t)}{\sqrt{t}},
    \]
    where \( G(T(t)) \sim \Gamma\left(\frac{T(t)}{\nu}, \nu\right) \) and \( T(t) \) is the ICIR process.

    Conditional on \( T(t) \), the mean and second moment of the VGSA process follow from properties of the Gamma distribution:
    \begin{align*}
        \mathbb{E}\left[VGSA(t) \mid T(t)\right] &= \theta\, \mathbb{E}\left[G(T(t)) \mid T(t)\right] = \theta\, T(t), \\
        \mathbb{E}\left[VGSA(t)^2 \mid T(t)\right] &= \theta^2\, \mathbb{E}\left[G(T(t))^2 \mid T(t)\right] + \sigma^2\, \mathbb{E}\left[G(T(t)) \mid T(t)\right] \\
        &= \theta^2\,\left(\nu\, T(t) + T(t)^2\right) + \sigma^2\, T(t).
    \end{align*}

    By taking expectations with respect to \( T(t) \) and utilizing the expressions provided in Lemma~\ref{lemma:icir_mean_var}, we derive:

    \begin{align*}
        \mathbb{E}[VGSA(t)] &= \theta\, \mathbb{E}[T(t)] = \theta\, u(t), \\
        \mathbb{E}[VGSA(t)^2] &= \theta^2\, \nu\, u(t) + \theta^2\, \mathbb{E}[T(t)^2] + \sigma^2\, u(t) \\
        &= \theta^2\, \nu\, u(t) + \theta^2\, \left(u(t)^2 + \lambda^2 w(t)\right) + \sigma^2\, u(t).
    \end{align*}

    Thus, the variance becomes:
    \begin{align*}
        \mathrm{Var}[VGSA(t)] 
        &= \mathbb{E}[VGSA(t)^2] - \left(\mathbb{E}[VGSA(t)]\right)^2 \\
        &= \theta^2 \nu\, u(t) + \sigma^2\, u(t) + \theta^2 \lambda^2 w(t) + \theta^2 u(t)^2 - \theta^2 u(t)^2 \\
        &= (\theta^2 \nu + \sigma^2) u(t) + \theta^2 \lambda^2 w(t).
    \end{align*}
    This completes the proof.
\end{proof}

We now extend this framework to analyze the key asymptotic properties of the VGSA process. Recall that the VG process can be interpreted as a Brownian motion subordinated by a Gamma process. By replacing the deterministic Gamma subordinator with a stochastic, almost surely increasing process—specifically the ICIR process—we obtain the VGSA process. Consider the CIR process~\eqref{eqn:cir_eqn} under the Feller condition \( 2\kappa\eta > \lambda^2 \), ensuring positivity. The VGSA process \eqref{eqn:vgsa_eqn} can equivalently be expressed as:
\begin{align*}
    VGSA(t) = VG(T(t)),
\end{align*}
where \( T(t) \) is the ICIR process~\eqref{eqn:icir_eqn}. The following lemma establishes the strong consistency of the VGSA process:

\begin{lemma}\label{lemma:vgsa_cir_as}
    The VGSA process satisfies the following almost sure convergence:
    \begin{align}
        \lim_{t \to \infty} \frac{1}{T(t)} VGSA(t) = \lim_{t \to \infty} \frac{1}{T(t)} VG(T(t)) \xrightarrow{a.s.} \theta. \label{eqn:VGSA_cir_as}
    \end{align}
\end{lemma}
More generally, we can establish the almost sure convergence of the VGSA process scaled by time, as follows:

\begin{lemma}\label{lemma:vgsa_as}
    The VGSA process satisfies the following strong law of large numbers-type asymptotic behavior:
    \begin{align}
        \lim_{t \to \infty} \frac{1}{t} VGSA(t) = \lim_{t \to \infty} \frac{1}{t} VG(T(t)) \xrightarrow{a.s.} \eta \cdot \theta. \label{eqn:VGSA_as}
    \end{align}
\end{lemma}

\begin{proof}
    From Lemma~\ref{lemma:vgsa_cir_as}, we have the almost sure convergence:
    \[
        \frac{VGSA(t)}{T(t)} = \frac{VG(T(t))}{T(t)} \xrightarrow{a.s.} \theta.
    \]
    Additionally, from Lemma~\ref{lemma:icir_as}, the integrated CIR process satisfies:
    \[
        \frac{T(t)}{t} \xrightarrow{a.s.} \eta.
    \]
    Combining these two results and applying the continuous mapping theorem yields:
    \begin{align*}
        \frac{1}{t} VGSA(t) = \frac{VG(T(t))}{T(t)} \cdot \frac{T(t)}{t} \xrightarrow{a.s.} \theta \cdot \eta,
    \end{align*}
    which completes the proof.
\end{proof}

We now establish the central limit theorem-type result for the VGSA process, characterizing its asymptotic distribution under appropriate scaling.

\begin{theorem}\label{thm:VGSA_normality}
    Consider VGSA~\eqref{eqn:vgsa_eqn}, with the stochastic clock $T(t)$ given by the integrated CIR (ICIR) process under the Feller condition \( 2\kappa\eta > \lambda^2 \). Then, as \( t \to \infty \), the following distributional convergence holds:
    \begin{align}
        \frac{1}{\sqrt{t}} \left(VGSA(t) - \eta \,\theta\, t \right) \distconv \mathcal{N} \left(0, \;\;\left( \theta^2\, \nu + \sigma^2 \right) \eta + \frac{\eta\, \theta^2\, \lambda^2}{\kappa^2} \right), \label{eqn:VGSA_dist_conv}
    \end{align}
    where $\distconv$ denotes convergence in distribution.
\end{theorem}


\subsection{Extending the Stochastic Clock: From CIR to CKLS}\label{subsec:ckls}
As previously discussed, the Stochastic Arrival (SA) component in the VGSA framework traditionally utilizes the Cox–Ingersoll–Ross (CIR) process to capture clustering effects in market activity. While effective, the CIR process is a special case of the more general CKLS model, which introduces an additional degree of flexibility by allowing for a nonlinear dependence of the diffusion term on the state variable. Importantly, the CKLS process retains positivity of the trajectories under mild regularity conditions (e.g., $\alpha > \frac{1}{2}$), making it a suitable candidate for modeling arrival intensities.

Given these advantages, it is natural to consider an extension of the VGSA framework by substituting the CIR clock with a CKLS process. This substitution enhances the model’s ability to capture more complex empirical features, provided that the corresponding integrated CKLS process remains almost surely increasing. Such an extension generalizes the VGSA architecture without compromising its foundational structure.

To proceed, we first present key properties of the CKLS process and its integrated form, which will be instrumental in deriving the asymptotic behavior and distributional properties of the extended VGSA model.

\begin{lemma}\label{lemma:ckls_stationary}
    Under the assumptions outlined above, the CKLS process defined by \eqref{eqn:ckls_diff_eqn} converges to a stationary distribution, and the boundary at zero is unattainable for all cases with $\alpha \in (0.5,1]$. For the special case $\alpha = \frac{1}{2}$, an additional condition $2\kappa\eta > \sigma^2$ is required to ensure strict positivity. The stationary density is given by:
    \begin{align}
        f(r) = C(\alpha) r^{-2\alpha} \exp(Q(r;\alpha)), \quad 
        C(\alpha)^{-1} = \int_0^\infty r^{-2\alpha} \exp(Q(r;\alpha))\, \mathrm{d}r, \label{eqn:stationary}
    \end{align}
    where $Q(r;\alpha)$ is defined as follows:
    \begin{align*}
        \text{For } \frac{1}{2} < \alpha < 1:\quad &
        Q(r;\alpha) = \frac{2\kappa}{\lambda^2} \left( \frac{\eta r^{1 - 2\alpha}}{1 - 2\alpha} - \frac{r^{2 - 2\alpha}}{2 - 2\alpha} \right), \\
        \text{For } \alpha = \frac{1}{2}:\quad &
        Q(r;\alpha) = \frac{2\kappa}{\lambda^2} \left( \eta \ln r - r \right), \\
        \text{For } \alpha = 1:\quad &
        Q(r;\alpha) = \frac{2\kappa}{\lambda^2} \left( -\frac{\eta}{r} - \ln r \right).
    \end{align*}
\end{lemma}

\begin{proof}
    See \citet{andersenPit} for a detailed derivation.
\end{proof}
To rigorously incorporate the CKLS process into our stochastic arrival (SA) framework, it is essential to establish its long-term behavior—specifically, its ergodic properties and the asymptotic distribution of the integrated process. These results not only parallel those derived for the CIR-based VGSA model but also provide the theoretical foundation necessary for extending the asymptotic analysis to more flexible SA-driven models such as CGMY-SA. In particular, we examine both strong consistency and asymptotic normality of the integrated CKLS process, which are critical for understanding the long-run behavior and distributional limits of time-changed Lévy models under CKLS dynamics.

\begin{theorem}\label{thm:ckls_strong_asymp}
    Let $\mathcal{Y}(t)$ be the integrated CKLS process defined by equations \eqref{eqn:ckls_diff_eqn} and \eqref{eqn:iergodic}. Then, under the stationarity and positivity assumptions for $r(t)$, we have the following asymptotic results:
    \begin{align}
        \lim_{t \to \infty} \frac{1}{t} \mathcal{Y}(t) \xrightarrow{a.s.} \eta, \label{eqn:ckls_strong}
    \end{align}
    and
    \begin{align}
        \frac{1}{\sqrt{t}}\left(\mathcal{Y}(t) - \eta t\right) \distconv \mathcal{N}\left(0, \frac{\lambda^2}{\kappa^2} \mathbb{E}\left[r^{2\alpha}\right]\right), \label{eqn:ckls_asymp}
    \end{align}
    where $r$ is a random variable following the stationary distribution specified in Lemma~\ref{lemma:ckls_stationary}.
\end{theorem}
In Theorem~\ref{thm:VGSA_normality}, we observed that the parameters $\kappa$ and $\lambda$ appear jointly in the variance term, leading to an identifiability issue when the CIR process is used as the stochastic clock. This challenge stems from the asymptotic properties of the integrated CIR process, as established in Lemma~\ref{lemma:ICIR_asymp}, where the limiting variance depends only on the ratio $\lambda^2/\kappa^2$. To mitigate this issue, one may consider replacing the CIR process with the more general CKLS process, defined by the stochastic differential equation, and choose
 $\alpha > 0.5$ ensures the well-posedness of the diffusion. In particular, choosing $\alpha = 1$ simplifies the analysis and yields an explicit expression for $\mathbb{E}[r^2]$, where $\kappa$ and $\lambda$ appear in a separable form. This formulation removes the identifiability ambiguity and enhances the interpretability of parameter estimates for inference and model calibration purposes.

Specifically, for $\alpha = 1$, the expectation can be written as:
\begin{align*}
    \mathbb{E}[r^2] = \frac{\int_0^\infty \exp(Q(r,1)) \, dr}{\int_0^\infty r^{-2} \exp(Q(r,1)) \, dr}.
\end{align*}
Evaluating the denominator gives:
\begin{align*}
    \int_0^\infty r^{-2} \left( \frac{1}{r} \right)^{\frac{2\kappa}{\lambda^2}} \exp\left(-\frac{2\kappa\eta}{\lambda^2} \cdot \frac{1}{r} \right) dr 
    &= \int_0^\infty z^{\frac{2\kappa}{\lambda^2}} \exp\left(-\frac{2\kappa\eta}{\lambda^2} z \right) dz \quad \left( z = \frac{1}{r} \right) \\
    &= \left( \frac{2\kappa\eta}{\lambda^2} \right)^{ -\left(\frac{2\kappa}{\lambda^2} + 1\right)} \Gamma\left( \frac{2\kappa}{\lambda^2} + 1 \right).
\end{align*}

Hence, provided that \( \frac{2\kappa}{\lambda^2} > 1 \), we ensure the finiteness of \( \mathbb{E}[r^2] \), which can be explicitly computed as:
\begin{align*}
    \mathbb{E}[r^2] &= \left( \frac{2\kappa\eta}{\lambda^2} \right)^2 \cdot \frac{\Gamma\left( \frac{2\kappa}{\lambda^2} - 1 \right)}{\Gamma\left( \frac{2\kappa}{\lambda^2} + 1 \right)} \\
    &= \left( \frac{2\kappa\eta}{\lambda^2} \right)^2 \cdot \frac{1}{\frac{2\kappa}{\lambda^2} \left( \frac{2\kappa}{\lambda^2} - 1 \right)} \\
    &= \frac{2\kappa \eta^2}{2\kappa - \lambda^2}.
\end{align*}
This expression reveals that with $\alpha = 1$, the identifiability issue between $\kappa$ and $\lambda$ can be disentangled, thereby offering a more tractable alternative to the CIR model in practice.

\section{Asymptotic Properties of Time-Changed L\'evy Processes: Beyond Gamma Subordinators}\label{sec:SBSA}

Building upon the extension of the VGSA framework through the CKLS-based stochastic arrival process discussed in sub-section \ref{subsec:ckls}, we now consider a further generalization of the jump component of the model. Specifically, instead of employing the Variance Gamma (VG) process, we focus on the CGMY process—a more flexible class of pure-jump Lévy processes that has been widely explored in the literature (see \citet{svlp}). 

An important feature of the CGMY process is that it admits a natural representation as a time-changed Brownian motion:
\begin{align}
        SB(t) = \theta S(t) + \sigma W(S(t)). \label{eqn:brow_sub}
    \end{align}
where \( W(t) \) denotes a standard Brownian motion, \( S(t) \) is a Lévy subordinator, \( \theta \) is the drift parameter, and \( \sigma \) is the diffusion coefficient. This structural form closely mirrors that of the VG process, offering a convenient framework for analysis and extending the modeling flexibility. In particular, this formulation allows the incorporation of various stochastic clocks, such as those driven by CKLS-type dynamics, thereby unifying jump modeling with stochastic arrival intensity.

Before delving into the full asymptotic behavior of the time-changed CGMY process under CKLS-type clocks, we first establish the fundamental asymptotic properties of generic time-changed Brownian motions with Lévy subordinators. These results will serve as a foundational stepping stone for analyzing the limiting behavior of more complex market microstructure models.

\begin{theorem}\label{thm:brow_sub_strong_const}
    Let \( S(t) \) be a Lévy subordinator such that \( \mathbb{E}[S(1)] < \infty \), and define the subordinated Brownian motion by the equation \eqref{eqn:brow_sub}.
    Then, the process \( SB(t) \) satisfies the following strong law of large numbers:
    \begin{align}
        \lim_{t \to \infty} \frac{SB(t)}{t} \xrightarrow{a.s.} \theta \cdot \mathbb{E}[S(1)]. \label{eqn:brow_sub_strong_cons}
    \end{align}
\end{theorem}

\begin{proof}
    Since \( S(t) \) is a Lévy subordinator with finite mean \( \mathbb{E}[S(1)] < \infty \), it follows from the strong law of large numbers for Lévy processes (see Theorem~\ref{thm:levy_strong}) that
    \[
        \lim_{t \to \infty} \frac{S(t)}{t} \xrightarrow{\text{a.s.}} \mathbb{E}[S(1)].
    \]

    Now consider the diffusion component \( W(S(t)) \). Using the scaling property of Brownian motion, we can write:
    \[
        W(S(t)) \overset{d}{=} \sqrt{S(t)} \cdot W(1).
    \]
    Therefore,
    \[
        \frac{W(S(t))}{t} \overset{d}{=} \frac{W(1) \cdot \sqrt{S(t)}}{t}.
    \]
    Since \( \sqrt{S(t)} = O(t^{1/2}) \) almost surely, it follows that
    \[
        \frac{W(S(t))}{t} \xrightarrow{\text{a.s.}} 0 \quad \text{as } t \to \infty.
    \]

    Combining both components, the subordinated process is given by:
    \[
        SB(t) = \theta S(t) + \sigma W(S(t)).
    \]
    Dividing through by \( t \), we have:
    \[
        \frac{SB(t)}{t} = \theta \cdot \frac{S(t)}{t} + \sigma \cdot \frac{W(S(t))}{t}.
    \]
    Using the limits established above, we obtain:
    \[
        \lim_{t \to \infty} \frac{SB(t)}{t} \xrightarrow{\text{a.s.}} \theta \cdot \mathbb{E}[S(1)] + 0 = \theta \cdot \mathbb{E}[S(1)],
    \]
    which establishes the result.
\end{proof}

To rigorously establish the asymptotic normality of subordinated Brownian motion processes, we impose a regularity condition on the Lévy measure \( v(dx) \) associated with the subordinator process \( S(t) \). Specifically, we require that the Lévy density exhibits sufficiently fast decay at infinity to ensure well-behaved moments:

\begin{description}
    \item[\textbf{Assumption A1:}] The Lévy measure \( v(dx) \) exhibits exponential decay, i.e.,
    \[
    \int_0^\infty e^x \, v(dx) < \infty.
    \]
\end{description}

This condition ensures that the second moment of the subordinator \( S(t) \) exists, which is crucial for proving the central limit behavior of the subordinated process.

\begin{theorem}\label{thm:SB_asymptotic}
    Let \( S(t) \) be a Lévy subordinator and let \( W(\cdot) \) be a standard Brownian motion independent of \( S(t) \). Define the subordinated Brownian motion as in equation \eqref{eqn:brow_sub} and 
    assume that~\textbf{A1} holds. Then, the process satisfies the following distributional convergence:
    \begin{align}
        \frac{1}{\sqrt{t}} \left( SB(t) - \theta \mathbb{E}[S(1)] t \right) \distconv \mathcal{N}\left(0, \sigma^2_{\mathrm{SB}}\right), \label{eqn:conv_subor}
    \end{align}
    where the asymptotic variance is given by
    \[
    \sigma^2_{\mathrm{SB}} = \sigma^2 \mathbb{E}[S(1)] + \theta^2 \int_{(0,\infty)} x^2 \, v(dx).
    \]
\end{theorem}

In \citet{svlp}, the authors introduce the framework of Stochastic Volatility Lévy Processes (SVLP), which captures market features such as jumps and volatility clustering through time-changed Lévy models. Rather than adopting the full generality of the SVLP structure, we focus on a specific and tractable subclass: the \emph{Subordinated Brownian motion with Stochastic Arrival} (SBSA) process, where the subordinated Brownian motion (SB) is defined by equation~\eqref{eqn:brow_sub}.

In this setup, we consider a time-changed construction in which the stochastic arrival is given by the integrated CKLS process. This approach retains the flexibility of the CKLS model—allowing for nonlinearity in the diffusion term—while preserving the monotonicity necessary for valid time changes. Concretely, we define a composite process where a Lévy subordinator is first evaluated at the integrated CKLS time, and the resulting process then serves as the input to a Brownian motion. This hierarchical structure, characteristic of SBSA models, is capable of capturing both temporal clustering in arrival times and jump behavior—features that are empirically observed in high-frequency financial data. The theorem below characterizes the asymptotic properties of models of this type.

\begin{theorem}\label{thm:SBSA_asymp}
Let the CKLS process be defined by the stochastic differential equation \eqref{eqn:ckls_diff_eqn}
where \( \alpha \in [0.5, 1] \). If \( \alpha = 0.5 \), we assume the Feller condition \( 2\kappa\eta > \lambda^2 \) to ensure positivity. Let the integrated CKLS process be given by equation \eqref{eqn:iergodic}
which is almost surely strictly increasing in \( t \). Now, let \( S(t) \) be a Lévy subordinator satisfying Assumption~\textbf{A1}, and let \( SB(t) \) be the subordinated Brownian motion defined in \eqref{eqn:brow_sub}. Define the time-changed process by:
\begin{align}
    SBSA(t) = \theta S(\mathcal{Y}(t)) + \sigma W(S(\mathcal{Y}(t))). \label{eqn:gen_vgtype}
\end{align}
Then, the process \( SBSA(t) \) satisfies the following asymptotic normality:
\begin{align}
    \frac{1}{\sqrt{t}} \left( SBSA(t) - \theta\,\eta\,\mathbb{E}[S(1)]\, t \right) \distconv \mathcal{N}(0, \sigma_1^2), \label{eqn:sbsa_all}
\end{align}
where the asymptotic variance \( \sigma_1^2 \) is given by
\[
\sigma_1^2 = \sigma_{SB}^2 \, \eta + \frac{\theta^2 \lambda^2}{\kappa^2} \mathbb{E}[r^{2\alpha}]\, \mathbb{E}[S(1)]^2,
\]
where $\sigma^2_{SB}$ defined from previous theorem \ref{thm:SB_asymptotic} and \( r \) denotes the stationary distribution of the CKLS process.
\end{theorem}
\textit{\bf Remark:} The process \( SBSA(t) \) can be interpreted as being generated by the composition (or convolution) \( W \circ S \circ \mathcal{Y}(t) \), representing a nested stochastic mechanism that integrates diffusion, jump and stochastic arrival components.

\section{Simulation Results}\label{sec:simu_res}

In this section, we present a simulation-based investigation to compare the long-term behavior of the Variance Gamma (VG) process and its time-changed counterpart, the Variance Gamma with Stochastic Arrival (VGSA) process. Specifically, we examine the scaled processes $\mathrm{VG}(t)/t$ and $\mathrm{VGSA}(t)/t$ across a range of time points and parameter configurations. The goal is to understand the influence of the CIR-based stochastic clock on the scaling behavior of the VG process.

Furthermore, we assess the asymptotic normality of these time-scaled processes and extend our simulation study to include more general time-changed structures, such as the Subordinated Brownian–Stochastic Arrival (SBSA) process, as introduced in Section~\ref{sec:SBSA}.

\subsection{Simulation related to VGSA}
We simulate both $\mathrm{VG}(t)/t$ and $\mathrm{VGSA}(t)/t$ over a grid of time values $t \in \{1, 2, \dots, 200\}$ for various parameter settings. This allows us to visually and quantitatively assess their convergence properties and to explore how the stochastic arrival mechanism driven by the CIR subordinator impacts the trajectory of the scaled process.

The following parameter combinations are used (see figure \ref{fig:VGSA_and_VG_wrt_t}):
\begin{figure}[!ht]
    \centering
    \includegraphics[width=0.8\linewidth]{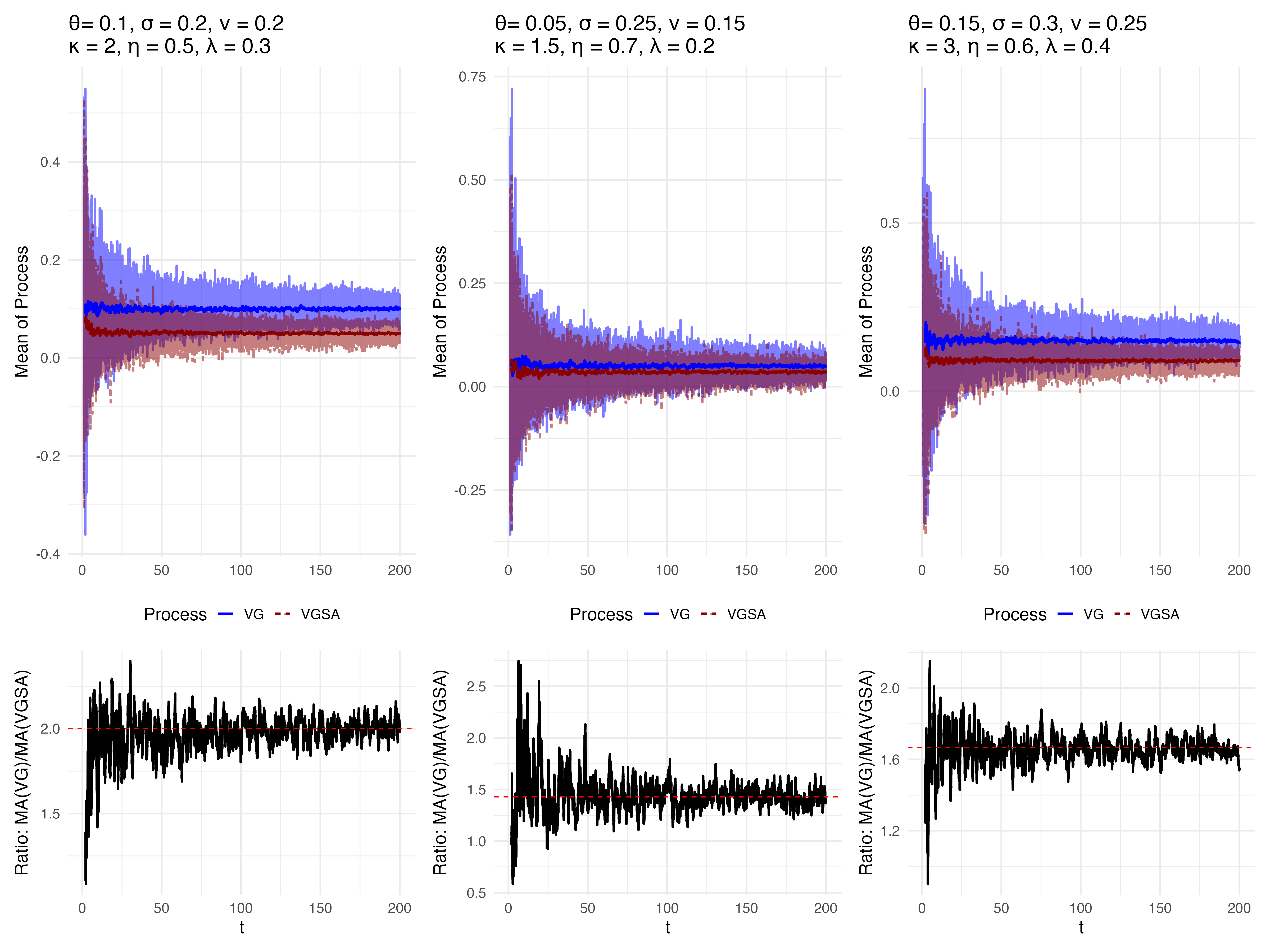}
    \caption{Empirical means of $\mathrm{VG}(t)/t$ and $\mathrm{VGSA}(t)/t$ plotted over $t \in \{1, 2, \ldots, 200\}$ for three different parameter sets. The VGSA process incorporates a stochastic time change via a CIR subordinator, which introduces additional nonlinearity and randomness. Each facet corresponds to a distinct set of parameters $(\theta, \sigma, \nu, \kappa, \eta, \lambda)$. Bottom row: Ratio of moving average (MA) of VG to VGSA sample paths, illustrating the relative volatility reduction introduced by the stochastic arrival (SA) mechanism. Each column corresponds to a different parameter configuration for $(\theta, \sigma, \nu, \kappa, \eta, \lambda)$. This is consistent with the results \ref{lemma:vg_as} and \ref{lemma:vgsa_as}.
}
    \label{fig:VGSA_and_VG_wrt_t}
\end{figure}

\begin{itemize}
    \item \textbf{Set P1:} $\theta = 0.1$, $\sigma = 0.2$, $\nu = 0.2$, $\kappa = 2.0$, $\eta = 0.5$, $\lambda = 0.3$
    \item \textbf{Set P2:} $\theta = 0.05$, $\sigma = 0.25$, $\nu = 0.15$, $\kappa = 1.5$, $\eta = 0.7$, $\lambda = 0.2$
    \item \textbf{Set P3:} $\theta = 0.15$, $\sigma = 0.3$, $\nu = 0.25$, $\kappa = 3.0$, $\eta = 0.6$, $\lambda = 0.4$
\end{itemize}

Each simulation is parallelized using the \texttt{parallel} package in \texttt{R} to improve computational efficiency. For every parameter configuration, we compute the empirical mean of both $\mathrm{VG}(t)/t$ and $\mathrm{VGSA}(t)/t$ over a range of time points. The results are visualized using line plots, faceted by parameter sets. Each plot illustrates the temporal evolution of the scaled processes, as discussed in Section~\ref{sec:VGSA}, particularly in Lemmas~\ref{lemma:vg_as} and~\ref{lemma:vgsa_as}. This facilitates a clear comparison of how the CIR-based stochastic clock influences the long-term behavior of the VG process under time change.

The VG process exhibits linear scaling in $t$, and hence $\mathrm{VG}(t)/t$ is expected to converge to a constant. In contrast, the VGSA process tends to exhibit lower variance due to the mean-reverting nature of the CIR subordinator, which regulates the variability of the stochastic clock. However, as $t$ becomes large, both $\mathrm{VG}(t)/t$ and $\mathrm{VGSA}(t)/t$ tend to stabilize, exhibiting approximately constant scaling with respect to $t$.

Next we confirm the normality of the VGSA process, which corresponds to theorem \ref{thm:VGSA_normality}. The simulation results (figure \ref{fig:VGSA_normality}) for the subordinated process $VGSA(T)$ across six different parameter configurations suggest that the resulting distributions are indeed approximately normal. This observation is supported by the Shapiro-Wilk test, where all p-values are well above the typical significance level of $0.05$, indicating no strong evidence against normality. Visually, the histograms and kernel density estimates further reinforce this, showing bell-shaped and symmetric patterns across all scenarios. Although minor deviations such as slight skewness or tail heaviness are observed in some sets—particularly where the VG drift parameter $\theta$ is non-zero or the CIR volatility is relatively high—the overall structure remains close to Gaussian. Hence, we conclude that under these settings, the subordinated VG process behaves more or less like a normally distributed random variable.


\begin{figure}[!ht]
    \centering
    \includegraphics[width=0.8\linewidth]{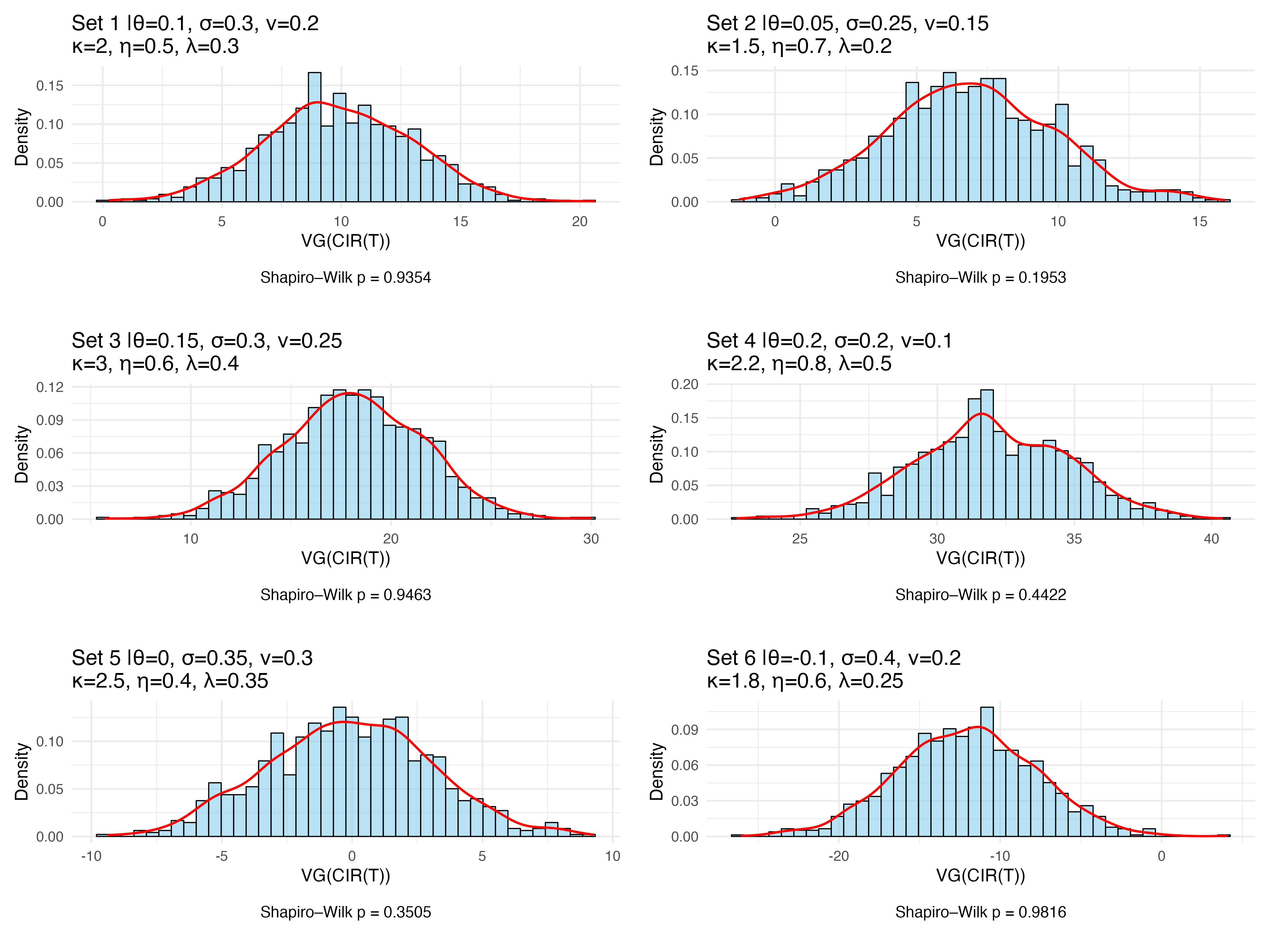}
    \caption{Histograms and kernel density estimates of the subordinated process $\mathrm{VG}(\mathrm{CIR}(T))$ under six different parameter settings. Each panel displays the empirical distribution based on 1000 Monte Carlo simulations, along with the corresponding Shapiro–Wilk $p$-value for normality testing. The results suggest that the distributions are approximately normal in all cases, consistent with Theorem~\ref{thm:VGSA_normality}.}

    \label{fig:VGSA_normality}
\end{figure}

\subsection{Simulation of SBSA Models}

We now shift our focus to the SBSA framework with more general Lévy processes. As discussed in Section~\ref{sec:SBSA}, the CGMY model \citep{cgmy} serves as a natural extension of the Variance Gamma (VG) process. The flexibility makes CGMY particularly well-suited for capturing complex empirical features of asset returns.

In this simulation study, we consider two extensions: the CGMY process time-changed by a CIR subordinator, denoted as $\mathrm{CGMY}(\mathrm{CIR}(t))$, and the CGMY process time-changed by a CKLS subordinator, denoted as $\mathrm{CGMY}(\mathrm{CKLS}(t))$. These models help us evaluate how different stochastic clocks influence the behavior of CGMY-type processes, particularly in the presence of time-varying market activity and volatility clustering.

\begin{figure}[!ht]
    \centering
    \includegraphics[width=0.8\linewidth]{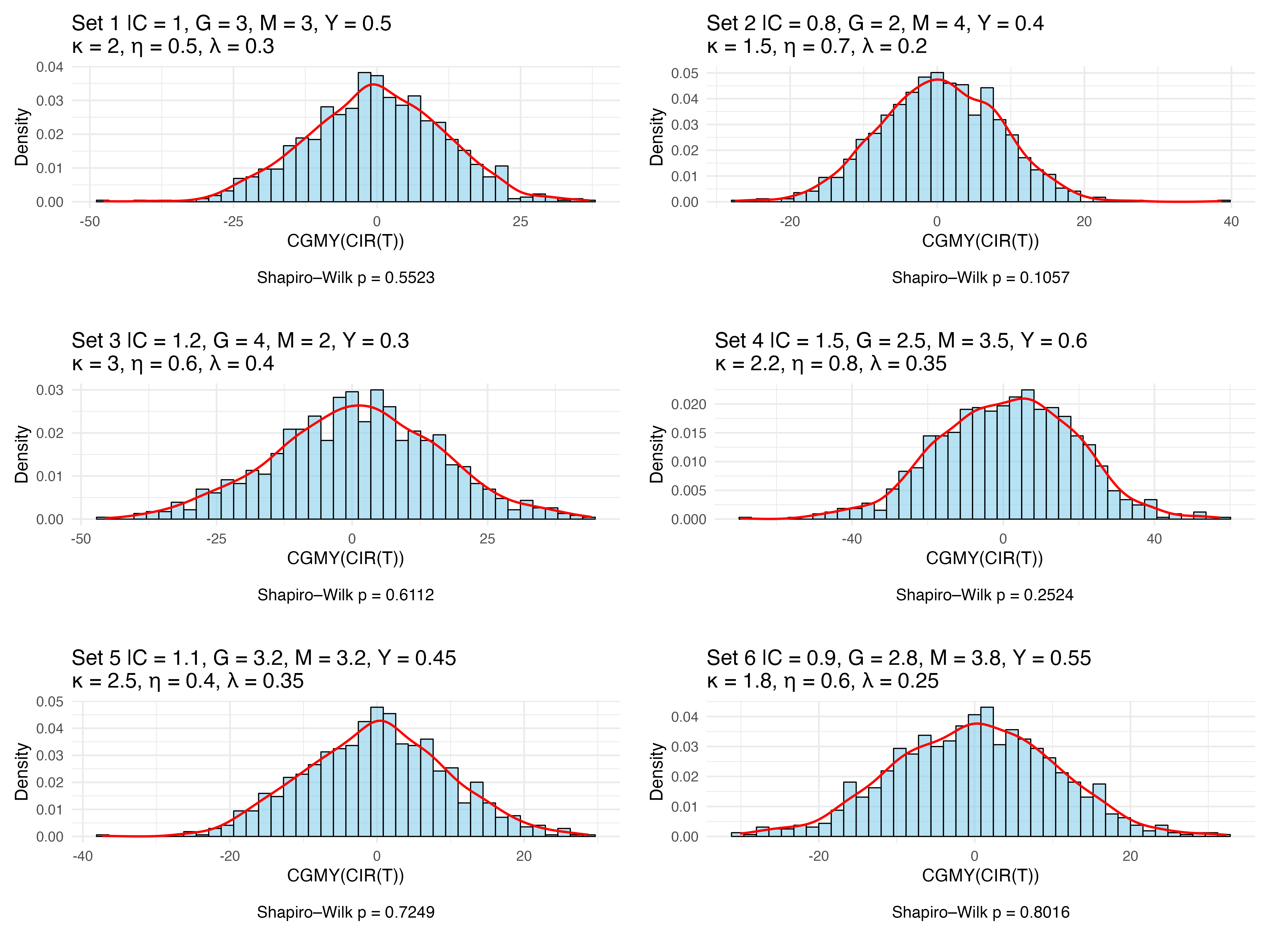}
    \caption{Histograms and kernel density estimates of the subordinated process $\mathrm{CGMY}(\mathrm{CIR}(T))$ under six different parameter settings. Each panel represents 1000 Monte Carlo samples, with the corresponding Shapiro-Wilk p-value displayed to assess normality (see theorem \ref{thm:SBSA_asymp}). Overall, the distributions appear approximately normal, though in some cases, slight deviations due to skewness or tail behavior are noticeable, reflecting the flexible jump structure of the CGMY model.}
    \label{fig:CGMYSA_normality}
\end{figure}

Figure~\ref{fig:CGMYSA_normality} presents the empirical distributions of the subordinated process $\mathrm{CGMY}(\mathrm{CIR}(T))$ across six different parameter configurations. Each subplot shows a histogram along with a kernel density estimate based on 1000 Monte Carlo simulations. The Shapiro–Wilk p-value is also reported for each configuration to assess normality. 

Across all parameter sets, the distributions exhibit bell-shaped and symmetric profiles, indicating that the subordinated process behaves approximately normally. The Shapiro–Wilk p-values range from 0.1953 to 0.9816, none of which fall below the conventional 5\% significance level. This supports the hypothesis that, although the CGMY process is inherently a pure-jump Lévy process, the CIR-based stochastic time change introduces sufficient smoothing to produce an approximately Gaussian outcome in finite samples.
Furthermore, Assumption~\textbf{A1} ensures that the tail of the Lévy measure decays exponentially, i.e., as $e^{-x}$, contributing to the asymptotic normality observed in these simulation results.

Among the six configurations, Sets 1, 3 and 6 exhibit particularly strong agreement with normality, as evidenced by their high Shapiro–Wilk p-values (0.9463 and 0.9816, respectively). In contrast, Set 2 yields a comparatively lower p-value (0.1953), which may be attributed to more pronounced jump activity or asymmetry in the Lévy measure arising from the specific choice of $Y$, $G$, and $M$ parameters. Overall, the simulation results suggest that subordinating CGMY processes with CIR-type stochastic clocks can effectively mitigate the heavy-tailed characteristics, resulting in distributions that closely approximate normality across a broad range of parameter settings.

\begin{figure}[!ht]
    \centering
    \begin{subfigure}[t]{0.7\linewidth}
        \centering
        \includegraphics[width=\linewidth]{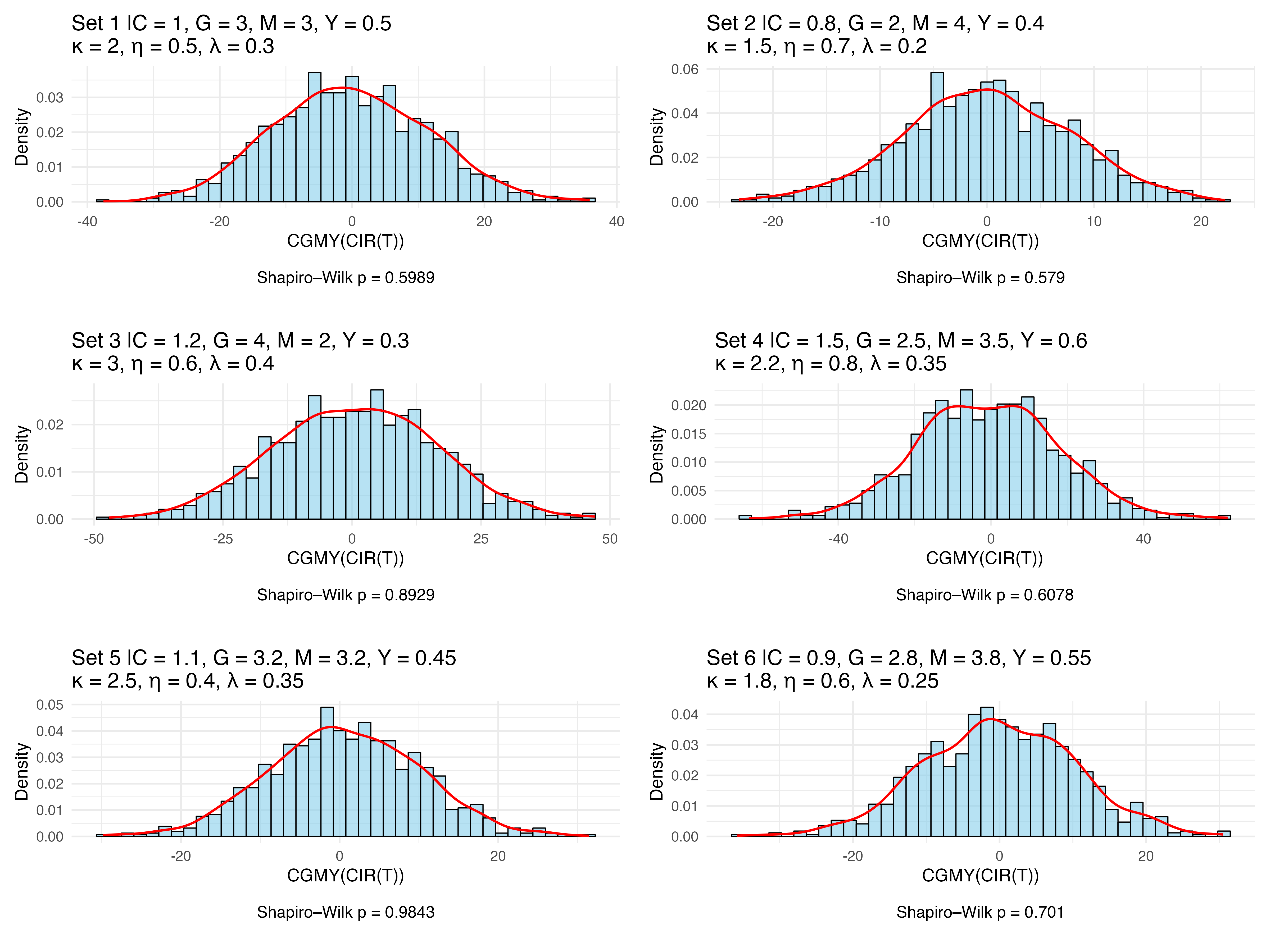}
        \caption{Subordinated process $\mathrm{CGMY}(\mathrm{CKLS}(T))$ with $\alpha = 0.6$.}
        \label{fig:cgmy_ckls_0_6}
    \end{subfigure}
    
    \vspace{1em} 
    
    \begin{subfigure}[t]{0.7\linewidth}
        \centering
        \includegraphics[width=\linewidth]{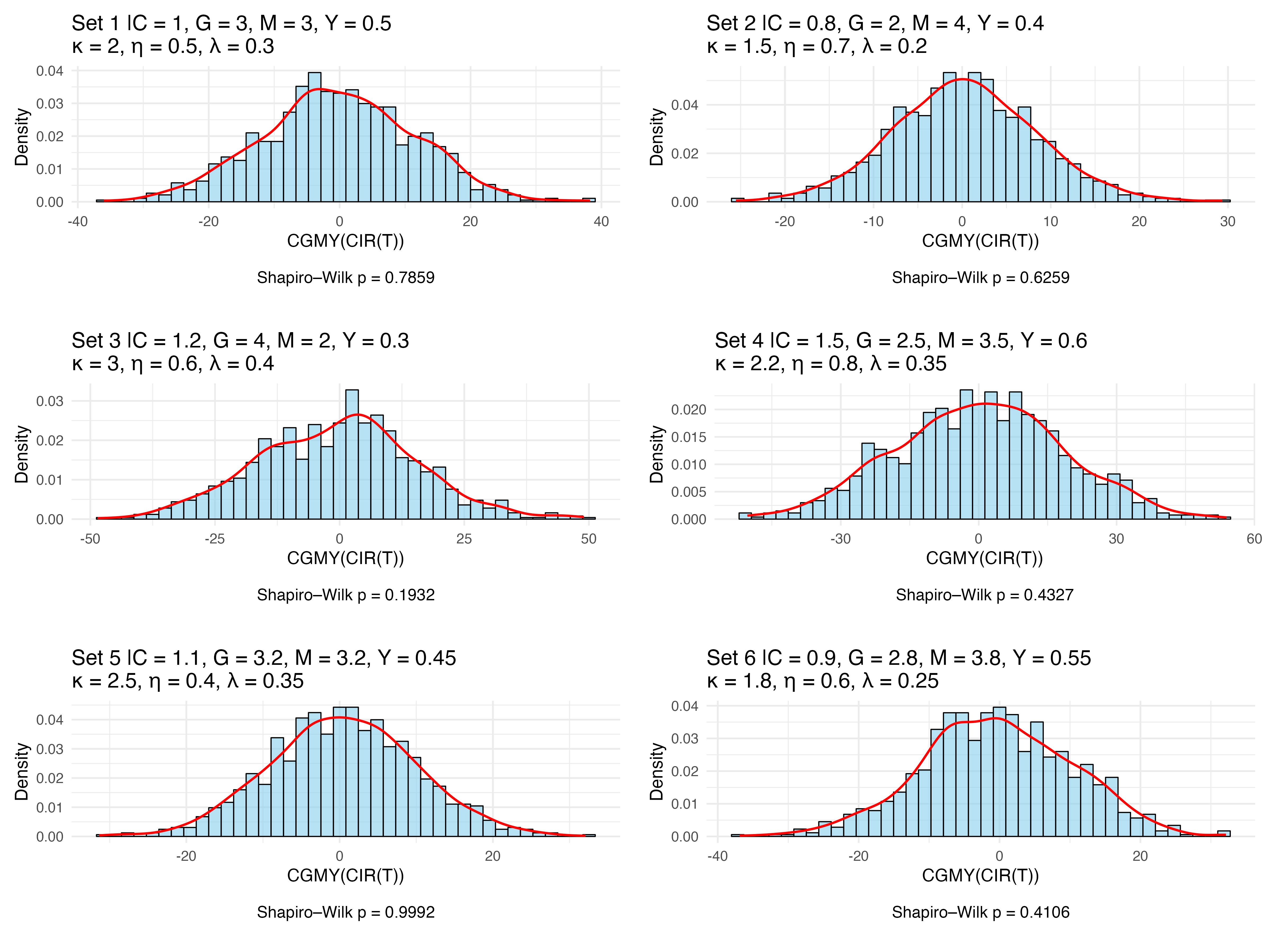}
        \caption{Subordinated process $\mathrm{CGMY}(\mathrm{CKLS}(T))$ with $\alpha = 1$.}
        \label{fig:cgmy_ckls_1}
    \end{subfigure}

    \caption{Histograms and kernel density estimates of the subordinated process $\mathrm{CGMY}(\mathrm{CKLS}(T))$ under six different parameter settings, for two values of the CKLS elasticity parameter, $\alpha \in {0.6, 1}$. Each panel is based on 1000 Monte Carlo simulations and includes the Shapiro–Wilk $p$-value for normality testing. The results indicate that the CKLS-based time change regularizes the CGMY process, yielding distributions that are approximately normal across various configurations, consistent with Theorem~\ref{thm:SBSA_asymp}.}
    \label{fig:cgmy_ckls_alpha}
\end{figure}

Similar to the earlier CIR-based analysis, we now explore the distributional properties of the subordinated CGMY process when the stochastic clock is governed by a CKLS process. Figures~\ref{fig:cgmy_ckls_0_6} and~\ref{fig:cgmy_ckls_1} illustrate the empirical behavior of $\mathrm{CGMY}(\mathrm{CKLS}(T))$ for $\alpha = 0.6$ and $\alpha = 1$, respectively, across six different parameter configurations. Each subplot includes a histogram with a kernel density estimate based on 1000 Monte Carlo simulations, accompanied by the corresponding Shapiro–Wilk $p$-value to assess normality.

For \textbf{$\alpha = 0.6$} (Figure~\ref{fig:cgmy_ckls_0_6}), the distributions exhibit smooth, symmetric, and unimodal shapes. Most $p$-values exceed the 5\% significance level, indicating no strong evidence against normality. In particular, Sets 3, 5, and 6 show strong agreement with Gaussian behavior, with Shapiro–Wilk $p$-values of 0.8929, 0.9843, and 0.701, respectively. Sets 1, 2, and 4 also maintain bell-shaped profiles with moderate $p$-values (0.5989, 0.579, and 0.6078), suggesting that the CKLS subordinator, even with $\alpha = 0.6$, effectively smooths the heavy-tailed structure of the original CGMY process.

Turning to \textbf{$\alpha = 1$} (Figure~\ref{fig:cgmy_ckls_1}), the normality approximation becomes even more evident. Set 5 achieves a near-perfect Gaussian fit with $p = 0.9992$, while Sets 1 and 3 also show favorable results ($p = 0.7859$ and $p = 0.1932$, respectively). Although Sets 2, 4, and 6 yield somewhat lower $p$-values (0.6259, 0.4327, and 0.4106), their empirical densities still resemble bell curves, suggesting only mild deviation from normality. These results underscore that increasing $\alpha$ strengthens the state-dependence in the CKLS diffusion component, yet does not disrupt the asymptotic regularization effect induced by the time change.

Taken together, the results demonstrate that CKLS-based subordination—regardless of the specific value of $\alpha$—is effective in smoothing and regularizing the distributional characteristics of the CGMY process. The emergence of Gaussian-like distributions across all parameter sets supports the theoretical expectation that ergodic stochastic clocks, when employed as time-change mechanisms, enhance the distributional stability of the subordinated process. These findings highlight the practical relevance of CKLS-driven stochastic arrival models for capturing empirical features of financial return data while preserving analytical tractability.

\section{Concluding Remarks}\label{sec:Fut_plan}

 In this paper, we developed a comprehensive framework for analyzing subordinated Brownian motion models where the subordination mechanism is either a classical Lévy subordinator or one governed by a stochastic arrival (SA) process, such as the CIR or CKLS diffusions. We began by reviewing the theory of Lévy processes and proved strong consistency for the Variance Gamma (VG) process under deterministic Gamma subordination. We then extended the analysis to time-changed Brownian motions, establishing both strong consistency and asymptotic normality for the VGSA process, where the time-change is driven by a stochastic process with positive support.

To generalize beyond the Gamma case, we introduced the CGMY process subordinated by SA processes, covering a wider class of pure-jump models. For both VGSA and CGMY-SA models, we rigorously derived sufficient conditions under which the subordinated processes remain consistent and converge in distribution to a normal limit. Special care was taken to ensure that the stochastic clocks (CIR or CKLS) satisfy positivity and ergodicity, which are crucial for the theoretical results to hold.
In continuation, we also defined and explored a new sub-class of the SVLP framework (\citet{svlp}), christened the \emph{Subordinated Brownian motion with Stochastic Arrival} (SBSA) process. We argue that this class incorporates three important features of market microstructure data - diffusion, jump and stochastic arrival - together, thus creating a very flexible model which also remains amenable to inference.   

Finally, we supported our theoretical findings with extensive simulation studies. Using Monte Carlo simulations, we visualized the distributional properties of VGSA and CGMY-SA under different parameter settings and SA types. The results, confirmed through histograms, kernel density estimates, and Shapiro-Wilk tests, consistently show that the time-changed processes exhibit approximately normal behavior across a wide range of configurations. This demonstrates that such subordinated models, despite their jump-driven structure, can be amenable to classical statistical inference when paired with appropriate stochastic arrival mechanisms.

These findings lay the groundwork for future exploration into parameter estimation, model calibration, and financial applications such as option pricing or volatility forecasting using subordinated Lévy models with realistic market features.

\section*{Appendix}\label{sec:app}

\begin{proof}[Proof of Theorem~\ref{theorem:subordi_as}]
    Let $X_t$ be a subordinator, i.e., an almost surely non-decreasing Lévy process. Define $\lfloor t \rfloor = n$ as the greatest integer less than or equal to $t$. Then we can decompose:
    \begin{align*}
        \frac{X_t}{t} = \frac{X_n}{n} \cdot \frac{n}{t} + \frac{X_t - X_n}{t}.
    \end{align*}

    The first term, $\frac{X_n}{n}$, is the average of $n$ i.i.d. increments of the Lévy process. By the strong law of large numbers,
    \[
        \frac{X_n}{n} \xrightarrow{a.s.} \mathbb{E}(X_1), \quad \text{as } n \to \infty,
    \]
    and since $\frac{n}{t} \to 1$ as $t \to \infty$, their product also converges almost surely to $\mathbb{E}(X_1)$.

    Now we analyze the second term. Since $X_t$ is non-decreasing,
    \begin{align*}
        0 \leq \frac{X_t - X_n}{t} \leq \sup_{n < u \leq n+1} \frac{X_u - X_n}{t} \leq \frac{X_{n+1} - X_n}{t}.
    \end{align*}

    Because the increments of $X_t$ are stationary and have finite mean, $X_{n+1} - X_n$ is independent of $t$ and identically distributed as $X_1$. Therefore,
    \[
        \frac{X_{n+1} - X_n}{t} \xrightarrow{a.s.} 0 \quad \text{as } t \to \infty.
    \]
    
    Combining both terms, we conclude that
    \[
        \frac{X_t}{t} \xrightarrow{a.s.} \mathbb{E}(X_1), \quad \text{as } t \to \infty.
    \]
    This completes the proof.
\end{proof}

\begin{proof}[Proof of Lemma~\ref{lemma:vg_as}]
    The Variance Gamma (VG) process can be represented as a Brownian motion with drift subordinated by a Gamma process:
    \[
        VG(t) = \theta G(t) + \sigma \sqrt{G(t)} \, W(t),
    \]
    where $G(t)$ is a Gamma subordinator with
    \[
        G(t) \sim \Gamma\left(\frac{t}{\nu}, \nu\right), \quad \text{so that } \mathbb{E}[G(t)] = t,
    \]
    and $W(t)$ is a standard Brownian motion independent of $G(t)$.

    Dividing both sides by $t$, we get:
    \[
        \frac{VG(t)}{t} = \theta \cdot \frac{G(t)}{t} + \sigma \cdot \frac{\sqrt{G(t)}}{\sqrt{t}} \cdot \frac{W(t)}{t}.
    \]

    From Theorem~\ref{theorem:subordi_as}, we have:
    \[
        \frac{G(t)}{t} \xrightarrow{a.s.} \mathbb{E}[G(1)] = 1, \quad \frac{\sqrt{G(t)}}{\sqrt{t}} \xrightarrow{a.s.} \sqrt{\mathbb{E}[G(1)]} = 1, \quad \frac{W(t)}{t} \xrightarrow{a.s.} 0.
    \]

    Therefore,
    \[
        \frac{VG(t)}{t} \xrightarrow{a.s.} \theta.
    \]

    This proves the almost sure convergence of the scaled VG process.
\end{proof}

\begin{proof}[Proof of Lemma~\ref{lemma:icir_mean_var}]
    The CIR process defined by equation~\eqref{eqn:cir_eqn} has the solution:
    \begin{align*}
        y(t) &= y(0) e^{-\kappa t} + \eta(1 - e^{-\kappa t}) + \lambda e^{-\kappa t} \int_0^t e^{\kappa u} \sqrt{y(u)} \, dW(u).
    \end{align*}
    Thus, the integrated CIR process is given by
    \begin{align}
        T(t) &= \int_0^t y(u) \, du \nonumber \\
        &= \int_0^t \left[ y(0) e^{-\kappa u} + \eta(1 - e^{-\kappa u}) + \lambda e^{-\kappa u} \int_0^u e^{\kappa z} \sqrt{y(z)} \, dW(z) \right] du \nonumber \\
        &= \int_0^t \left[ y(0) e^{-\kappa u} + \eta(1 - e^{-\kappa u}) \right] du + \lambda \int_0^t e^{-\kappa u} \left[ \int_0^u e^{\kappa z} \sqrt{y(z)} \, dW(z) \right] du \nonumber \\
        &= -\frac{y(0)}{\kappa}(e^{-\kappa t} - 1) + \eta \left( t + \frac{1}{\kappa}(e^{-\kappa t} - 1) \right) + \lambda \int_0^t e^{-\kappa u} \left[ \int_0^u e^{\kappa z} \sqrt{y(z)} \, dW(z) \right] du \nonumber \\
        &= u(t) + \lambda I(t), \label{eqn:decom_icir}
    \end{align}
    where
    \begin{align*}
        u(t) &= -\frac{y(0)}{\kappa}(e^{-\kappa t} - 1) + \eta \left( t + \frac{1}{\kappa}(e^{-\kappa t} - 1) \right), \\
        I(t) &= \int_0^t e^{-\kappa u} \left[ \int_0^u e^{\kappa z} \sqrt{y(z)} \, dW(z) \right] du.
    \end{align*}

    We now analyze $I(t)$. Consider a pure diffusion process $\mathcal{W}(t)$; then:
    \begin{align*}
        d\left(e^{-\kappa t} \mathcal{W}(t)\right) &= e^{-\kappa t} \, d\mathcal{W}(t) - \kappa \mathcal{W}(t) e^{-\kappa t} \, dt.
    \end{align*}
    Integrating both sides and rearranging, we obtain:
    \begin{align*}
        \int_0^t \mathcal{W}(u) e^{-\kappa u} \, du &= \frac{1}{\kappa} \left( \int_0^t e^{-\kappa u} \, d\mathcal{W}(u) - e^{-\kappa t} \mathcal{W}(t) + \mathcal{W}(0) \right) \\
        &= \frac{1}{\kappa} \int_0^t (e^{-\kappa u} - e^{-\kappa t}) \, d\mathcal{W}(u).
    \end{align*}

    Substituting $\mathcal{W}(t) = \int_0^t e^{\kappa u} \sqrt{y(u)} \, dW(u)$, we get:
    \begin{align*}
        I(t) &= \frac{1}{\kappa} \int_0^t (e^{-\kappa u} - e^{-\kappa t}) e^{\kappa u} \sqrt{y(u)} \, dW(u) \\
        &= \frac{e^{-\kappa t}}{\kappa} \int_0^t \left( e^{\kappa t} - e^{\kappa u} \right) \sqrt{y(u)} \, dW(u).
    \end{align*}

    Thus, $I(t)$ is a pure diffusion process and satisfies $\mathbb{E}[I(t)] = 0$. Therefore, the expectation of $T(t)$ is:
    \begin{align*}
        \mathbb{E}[T(t)] &= \mathbb{E}[u(t) + \lambda I(t)] = u(t),
    \end{align*}
    which depends explicitly on the initial value $y(0)$.

    Next, we compute $\mathbb{E}[I(t)^2]$ using Itô isometry:
    \begin{align*}
        \mathbb{E}[I(t)^2] &= \frac{1}{\kappa^2} e^{-2\kappa t} \mathbb{E} \left[ \int_0^t (e^{\kappa t} - e^{\kappa z})^2 y(z) \, dz \right] \\
        &= \frac{1}{\kappa^2} e^{-2\kappa t} \int_0^t (e^{\kappa t} - e^{\kappa z})^2 \mathbb{E}[y(z)] \, dz \\
        &= \frac{1}{\kappa^2} e^{-2\kappa t} \int_0^t (e^{\kappa t} - e^{\kappa z})^2 \left[ y(0) e^{-\kappa z} + \eta (1 - e^{-\kappa z}) \right] dz.
    \end{align*}
    Splitting and simplifying:
    \begin{align*}
        \mathbb{E}[I(t)^2] &= \frac{y(0)}{\kappa^2} e^{-2\kappa t} \int_0^t (e^{\kappa t} - e^{\kappa z})^2 e^{-\kappa z} \, dz \\
        &\quad + \frac{\eta}{\kappa^2} e^{-2\kappa t} \int_0^t (e^{\kappa t} - e^{\kappa z})^2 (1 - e^{-\kappa z}) \, dz \\
        &=: w(t).
    \end{align*}

    Therefore, the second moment of $T(t)$ is:
    \begin{align*}
        \mathbb{E}[T(t)^2] &= u(t)^2 + \lambda^2 \mathbb{E}[I(t)^2],
    \end{align*}
    and since $\mathbb{E}[u(t) I(t)] = 0$, the variance simplifies to:
    \begin{align*}
        \operatorname{Var}(T(t)) &= \mathbb{E}[T(t)^2] - \mathbb{E}[T(t)]^2 = \lambda^2 w(t).
    \end{align*}
    This completes the proof.
\end{proof}

\begin{proof}[Proof of Lemma~\ref{lemma:vgsa_cir_as}]
    Following a similar approach as in the previous lemma, let us consider the greatest integer function of the stochastic time change, denoted by $n = \lfloor T(t) \rfloor$. Then we can write:
    \begin{align*}
        \frac{\mathrm{VGSA}(t)}{T(t)} &= \frac{\mathrm{VG}(T(t))}{T(t)} \\
        &= \frac{\mathrm{VG}(n) - \mathrm{VG}(n) + \mathrm{VG}(T(t))}{T(t)} \\
        &= \frac{\mathrm{VG}(n)}{n} \cdot \frac{n}{T(t)} + \frac{\mathrm{VG}(T(t)) - \mathrm{VG}(n)}{T(t)}.
    \end{align*}
    From the earlier result on the almost sure behavior of the VG process, we know that
    \[
        \frac{\mathrm{VG}(n)}{n} \xrightarrow{a.s.} \theta,
    \]
    and also $\frac{n}{T(t)} \to 1$ almost surely as $t \to \infty$. This follows from the inequality:
    \[
        \frac{n - 1}{T(t)} \leq \frac{n}{T(t)} \leq \frac{T(t)}{T(t)} = 1,
    \]
    and the fact that $T(t) \to \infty$ almost surely (since it is the integral of a positive ergodic process). Hence,
    \[
        \frac{n}{T(t)} \xrightarrow{a.s.} 1.
    \]

    It remains to show that the second term vanishes almost surely as $t \to \infty$. As in the earlier lemma, we write:
    \begin{align*}
        \left| \frac{\mathrm{VG}(T(t)) - \mathrm{VG}(n)}{T(t)} \right|
        &\leq \sup_{n < u \leq n+1} \frac{|\mathrm{VG}(u) - \mathrm{VG}(n)|}{t} \cdot \frac{t}{T(t)} \\
        &= \underbrace{\frac{\mathrm{VG}(n+1) - \mathrm{VG}(n)}{t}}_{\xrightarrow{a.s.} 0} \cdot \underbrace{\frac{t}{T(t)}}_{\xrightarrow{a.s.} \eta^{-1}} \\
        &\xrightarrow{a.s.} 0.
    \end{align*}

    Therefore,
    \[
        \frac{\mathrm{VGSA}(t)}{T(t)} \xrightarrow{a.s.} \theta.
    \]
    This completes the proof.
\end{proof}

\begin{proof}[Proof of Theorem~\ref{thm:VGSA_normality}]
    We prove the result using the characteristic function of the VGSA process. Note that:
    \begin{align*}
        \log \mathbb{E}\left[e^{iu\, \mathrm{VGSA}(t)} \,\middle|\, T(t)\right]
        &= -\frac{T(t)}{\nu} \log\left(1 - iu\theta\nu + \frac{\sigma^2 u^2 \nu}{2} \right)
        = T(t) s,
    \end{align*}
    where
    \[
        s = -\frac{1}{\nu} \log\left(1 - iu\theta\nu + \frac{\sigma^2 u^2 \nu}{2} \right).
    \]
    Since $T(t)/t \xrightarrow{a.s.} \eta$ as $t \to \infty$, it follows that:
    \[
        \frac{1}{t} \log \mathbb{E}\left[e^{iu\, \mathrm{VGSA}(t)} \,\middle|\, T(t)\right] \xrightarrow{a.s.} -\frac{\eta}{\nu} \log\left(1 - iu\theta\nu + \frac{\sigma^2 u^2 \nu}{2} \right).
    \]

    Next, consider the moment generating function of $T(t)$:
    \[
        \mathbb{E}\left[e^{s T(t)}\right] \approx \exp\left(s \eta t + \frac{1}{2} s^2 \frac{\eta \lambda^2}{\kappa^2} t \right),
    \]
    using the moment-generating approximation of the CIR time change.

    Now, substitute $u \mapsto u/\sqrt{t}$ in the characteristic exponent and expand $\log(1 + x)$ using a second-order Taylor expansion:
    \begin{align*}
        s &= -\frac{1}{\nu} \log\left(1 - \frac{iu\theta\nu}{\sqrt{t}} + \frac{\sigma^2 u^2 \nu}{2t} \right) \\
        &\approx \frac{1}{\nu} \left( \frac{iu\theta\nu}{\sqrt{t}} - \frac{\sigma^2 u^2 \nu}{2t} + \frac{(iu\theta\nu)^2}{2t} + o\left(\frac{1}{t}\right) \right) \\
        &= \frac{iu\theta}{\sqrt{t}} - \frac{u^2}{2t} \left( \sigma^2 + \nu \theta^2 \right) + o\left(\frac{1}{t} \right).
    \end{align*}

    Similarly, expanding $s^2$ gives:
    \[
        s^2 \approx -\frac{u^2 \theta^2}{t} + o\left(\frac{1}{t} \right).
    \]

    Substituting back into the moment-generating expression, we obtain:
    \begin{align*}
        \mathbb{E}\left[e^{i \frac{u}{\sqrt{t}} \mathrm{VGSA}(t)}\right]
        &\approx \exp\left\{ t \left( \frac{iu\theta \eta}{\sqrt{t}} - \frac{u^2}{2t} \left( \sigma^2 \eta + \nu \theta^2 \eta \right) \right)
        + \frac{1}{2} t \cdot \left( -\frac{u^2 \theta^2 \eta \lambda^2}{\kappa^2 t} \right)
        + o\left(\frac{1}{t} \right) \right\} \\
        &= \exp\left\{ iu \sqrt{t}\, \eta \theta - \frac{u^2}{2} \left( \sigma^2 \eta + \nu \theta^2 \eta + \frac{\lambda^2 \theta^2 \eta}{\kappa^2} \right)
        + o\left( \frac{1}{t} \right) \right\}.
    \end{align*}

    Therefore, the characteristic function of the centered and scaled process,
    \[
        \frac{1}{\sqrt{t}} \left( \mathrm{VGSA}(t) - \eta \theta t \right),
    \]
    converges pointwise to that of a normal distribution with mean $0$ and variance
    \[
        \sigma^2 \eta + \nu \theta^2 \eta + \frac{\lambda^2 \theta^2 \eta}{\kappa^2}.
    \]
    Hence, we have shown convergence in distribution to a Gaussian limit.
\end{proof}

\begin{proof}[Proof of Theorem~\ref{thm:ckls_strong_asymp}]
    The integrated CKLS process defined in equation~\eqref{eqn:iergodic} satisfies the following:
    \begin{align*}
        \mathcal{Y}(t) &= \int_0^t \left[ y(0) e^{-\kappa u} + \eta (1 - e^{-\kappa u}) + \lambda e^{-\kappa u} \int_0^u e^{\kappa z} y(z)^\alpha \, dW(z) \right] du \\
        &= \frac{y(0)}{\kappa}(1 - e^{-\kappa t}) + \eta \left(t + \frac{1}{\kappa}(e^{-\kappa t} - 1) \right) + \lambda \int_0^t \int_0^u e^{-\kappa(u - z)} y(z)^\alpha \, dW(z) \, du \\
        &=: \mathcal{U}(t) + \lambda \mathcal{I}(t).
    \end{align*}

    To study the asymptotic behavior, we focus on the second term:
    \[
        \mathcal{I}(t) = \int_0^t \left( \frac{1}{\kappa} e^{\kappa z} (1 - e^{-\kappa(t - z)}) \right) y(z)^\alpha \, dW(z).
    \]

    Applying Itô isometry, we obtain:
    \begin{align*}
        \mathbb{E}[\mathcal{I}(t)^2] &= \frac{1}{\kappa^2} \int_0^t \left( 1 - e^{-\kappa(t - z)} \right)^2 y(z)^{2\alpha} \, dz \\
        &= \frac{1}{\kappa^2} \int_0^t y(z)^{2\alpha} \, dz - \frac{2}{\kappa^2} \int_0^t e^{-\kappa(t - z)} y(z)^{2\alpha} \, dz + \frac{1}{\kappa^2} \int_0^t e^{-2\kappa(t - z)} y(z)^{2\alpha} \, dz \\
        &=: \mathcal{I}_1(t) - \mathcal{I}_2(t) + \mathcal{I}_3(t).
    \end{align*}

    From the ergodic property of the CKLS process, if $\mathbb{E}[r^{2\alpha}]$ exists (where $r$ follows the stationary distribution described in Lemma~\ref{lemma:ckls_stationary}), then:
    \[
        \frac{1}{t} \mathcal{I}_1(t) \xrightarrow{a.s.} \frac{1}{\kappa^2} \mathbb{E}[r^{2\alpha}].
    \]

    For $\mathcal{I}_2(t)$, using the ergodic theorem and boundedness of the exponential decay:
    \begin{align*}
        \mathcal{I}_2(t) &= \frac{2}{\kappa^2} \int_0^t e^{-\kappa(t - z)} y(z)^{2\alpha} \, dz \\
        &= \frac{2}{\kappa^2} \left( \int_0^t e^{-\kappa(t - z)} \left( y(z)^{2\alpha} - \mathbb{E}[r^{2\alpha}] \right) dz + \mathbb{E}[r^{2\alpha}] \int_0^t e^{-\kappa(t - z)} dz \right) \\
        &\le \frac{2}{\kappa^2} \left( \int_0^t \left| y(z)^{2\alpha} - \mathbb{E}[r^{2\alpha}] \right| dz + \mathbb{E}[r^{2\alpha}] \right),
    \end{align*}
    implying
    \[
        \frac{1}{t} \mathcal{I}_2(t) \xrightarrow{a.s.} 0.
    \]

    A similar way applies to $\mathcal{I}_3(t)$:
    \[
        \frac{1}{t} \mathcal{I}_3(t) \xrightarrow{a.s.} 0.
    \]

    Combining the limits of all three terms:
    \[
        \frac{1}{t} \mathcal{I}(t)^2 \xrightarrow{a.s.} \frac{1}{\kappa^2} \mathbb{E}[r^{2\alpha}].
    \]

    Therefore, by the Martingale Central Limit Theorem (CLT), we obtain:
    \[
        \frac{1}{\sqrt{t}} \mathcal{I}(t) \distconv \mathcal{N}\left(0, \frac{1}{\kappa^2} \mathbb{E}[r^{2\alpha}] \right).
    \]

    Since $\mathcal{Y}(t) = \mathcal{U}(t) + \lambda \mathcal{I}(t)$, we analyze the centered and scaled version:
    \[
        \frac{1}{\sqrt{t}} \left( \mathcal{Y}(t) - \eta t \right) = \frac{\lambda}{\sqrt{t}} \mathcal{I}(t) + \frac{1}{\sqrt{t}} \left( \mathcal{U}(t) - \eta t \right).
    \]
    The first term converges in distribution by the Martingale CLT, and the second term converges almost surely to 0. Hence, by Slutsky's theorem:
    \[
        \frac{1}{\sqrt{t}} \left( \mathcal{Y}(t) - \eta t \right) \distconv \mathcal{N}\left(0, \frac{\lambda^2}{\kappa^2} \mathbb{E}[r^{2\alpha}] \right).
    \]
\end{proof}

\begin{proof}[Proof of Theorem~\ref{thm:SB_asymptotic}]
    We begin by analyzing the characteristic function of the $SB(t)$ process. By definition,
    \begin{align*}
        \mathbb{E}\left[e^{iu \, SB(t)}\right] 
        &= \mathbb{E}\left[e^{iu\theta S(t) - \frac{1}{2} u^2 \sigma^2 S(t)}\right] \\
        &= \mathbb{E}\left[e^{\left(iu\theta - \frac{1}{2} u^2 \sigma^2\right) S(t)}\right].
    \end{align*}
    Let $u' := iu\theta - \frac{1}{2} u^2 \sigma^2$. Then, using the Lévy–Khintchine formula for the Laplace exponent of the subordinator $S(t)$, we obtain:
    \begin{align*}
        \log \mathbb{E}\left[e^{u' S(t)}\right] = t \left( u' \gamma + \int_0^\infty \left(e^{u' x} - 1\right) v(dx) \right),
    \end{align*}
    where $\gamma$ is the drift coefficient and $v$ is the Lévy measure of the subordinator.

    Hence, the characteristic function of $SB(t)$ is:
    \begin{align*}
        \mathbb{E}\left[e^{iu \, SB(t)}\right] 
        = \exp\left\{ t \left( u' \gamma + \int_0^\infty \left( e^{u' x} - 1 \right) v(dx) \right) \right\}.
    \end{align*}

    To analyze the asymptotic behavior, substitute $u$ with $u/\sqrt{t}$. Then:
    \begin{align*}
        u' &= i \frac{u}{\sqrt{t}} \theta - \frac{1}{2} \frac{u^2}{t} \sigma^2.
    \end{align*}
    Applying the second-order Taylor expansion of $e^{u'x}$ around $u'=0$, we get:
    \begin{align*}
        e^{u'x} - 1 
        = u' x + \frac{1}{2} u'^2 x^2 + R\left(\tfrac{1}{t}\right),
    \end{align*}
    where $R\left(\tfrac{1}{t}\right)$ denotes a remainder term such that $R\left(\tfrac{1}{t}\right) \to 0$ as $t \to \infty$, under assumption \textbf{A1}.

    Substituting into the expression for the log-characteristic function:
    \begin{align*}
        \log \mathbb{E}\left[e^{iu \, SB(t)}\right] 
        &= t \left( u' \gamma + \int_0^\infty \left(u' x + \frac{1}{2} u'^2 x^2 + R\left(\tfrac{1}{t}\right) \right) v(dx) \right) \\
        &= u' t \left( \gamma + \int_0^\infty x v(dx) \right) + \frac{1}{2} u'^2 t \int_0^\infty x^2 v(dx) + t \cdot R\left(\tfrac{1}{t}\right).
    \end{align*}

    Expanding the terms:
    \begin{align*}
        \log \mathbb{E}\left[e^{iu \, SB(t)}\right]
        &= i u \sqrt{t} \theta \left( \gamma + \int_0^\infty x v(dx) \right)
        - \frac{1}{2} u^2 \left( \sigma^2 \gamma + \sigma^2 \int_0^\infty x v(dx) + \theta^2 \int_0^\infty x^2 v(dx) \right) \\
        &\quad + R\left(\tfrac{1}{t}\right).
    \end{align*}

    Noting that $\mathbb{E}[S(1)] = \gamma + \int_0^\infty x v(dx)$, we define the limiting variance:
    \[
        \sigma_{sb}^2 := \sigma^2 \mathbb{E}[S(1)] + \theta^2 \int_0^\infty x^2 v(dx).
    \]

    Therefore, we conclude:
    \[
        \log \mathbb{E}\left[e^{iu \, SB(t)}\right] 
        \approx i u \sqrt{t} \theta \mathbb{E}[S(1)] - \frac{1}{2} u^2 \sigma_{SB}^2 + o(1),
    \]
    which implies the following convergence in distribution:
    \[
        \frac{SB(t) - \theta \mathbb{E}[S(1)] t}{\sqrt{t}} \distconv \mathcal{N}(0, \sigma_{SB}^2).
    \]
\end{proof}

\begin{proof}[Proof of Theorem~\ref{thm:SBSA_asymp}]
    For the model $SBSA(t)$, the characteristic function is given by
    \begin{align*}
        \mathbb{E}\left[\exp\left(iu\, SBSA(t)\right)\right] 
        = \mathbb{E}\left[\exp\left(\mathcal{Y}(t) \cdot \left( u'\gamma + \int_0^\infty \left(e^{u'x} - 1\right) v(dx) \right)\right)\right],
    \end{align*}
    where $u' := iu\theta - \frac{1}{2} u^2 \sigma^2$.

    Replacing $u$ by $\frac{u}{\sqrt{t}}$ and using the asymptotic expansion of $\mathcal{Y}(t)$ from Theorem~\ref{thm:ckls_strong_asymp}, we obtain
    \begin{align*}
        \mathbb{E}\left[\exp\left(iu\, SBSA(t)\right)\right] 
        \approx \exp\left\{ \mathcal{S}\,\eta\, t + \frac{1}{2} \mathcal{S}^2 \cdot \frac{\lambda^2}{\kappa^2} \cdot \mathbb{E}(r^{2\alpha}) \, t \right\},
    \end{align*}
    where
    \[
        \mathcal{S} := u'\gamma + \int_0^\infty \left(e^{u'x} - 1\right) v(dx), \quad 
        \text{and} \quad u' = i\frac{u}{\sqrt{t}}\theta - \frac{1}{2} \frac{u^2}{t} \sigma^2.
    \]

    We analyze the exponent term by term. First,
    \begin{align*}
        \mathcal{S}\,\eta\,t 
        &\approx \eta \left[ iu\sqrt{t}\theta\, \mathbb{E}[S(1)] 
        - \frac{1}{2} u^2 \left( \sigma^2\, \mathbb{E}[S(1)] + \theta^2 \int_0^\infty x^2 v(dx) \right) \right] 
        + o\left(\tfrac{1}{t}\right).
    \end{align*}

    For the second-order term:
    \begin{align*}
        \frac{1}{2} \mathbb{E}[r^{2\alpha}]\, \frac{\lambda^2}{\kappa^2}\, t\, \mathcal{S}^2 
        &= \frac{1}{2} \mathbb{E}[r^{2\alpha}]\, \frac{\lambda^2}{\kappa^2}\, t 
        \left[ -\frac{u^2}{t} \theta^2 \left( \gamma + \int_0^\infty x v(dx) \right)^2 + \mathcal{R}\left(\tfrac{1}{t}\right) \right] \\
        &= -\frac{u^2}{2} \cdot \frac{\lambda^2}{\kappa^2} \cdot \theta^2\, \mathbb{E}[r^{2\alpha}]\, \mathbb{E}[S(1)]^2 + o(1),
    \end{align*}
    where $t\mathcal{R}(1/t) := o(1)$.

    Combining both parts, we obtain the expansion:
    \begin{align*}
        \log \mathbb{E}\left[\exp\left(iu\, SBSA(t)\right)\right]
        &= iu\sqrt{t}\, \eta \theta\, \mathbb{E}[S(1)] \\
        &\quad - \frac{u^2}{2} \left( \eta \sigma^2\, \mathbb{E}[S(1)] 
        + \eta \theta^2 \int_0^\infty x^2 v(dx) 
        + \theta^2 \frac{\lambda^2}{\kappa^2} \mathbb{E}[r^{2\alpha}]\, \mathbb{E}[S(1)]^2 \right) + o(1).
    \end{align*}

    Define:
    \begin{align*}
        \sigma^2_{SB} &= \sigma^2\, \mathbb{E}[S(1)] + \theta^2 \int_0^\infty x^2 v(dx), \\
        \sigma^2_1 &= \eta\, \sigma^2_{SB} + \theta^2\, \frac{\lambda^2}{\kappa^2}\, \mathbb{E}[r^{2\alpha}]\, \mathbb{E}[S(1)]^2.
    \end{align*}

    Therefore, we conclude:
    \begin{align*}
        \frac{SBSA(t) - t \eta \theta\, \mathbb{E}[S(1)]}{\sqrt{t}} 
        \distconv \mathcal{N}(0, \sigma^2_1).
    \end{align*}
\end{proof}

\printbibliography

\end{document}